%% file: triangle_arxiv.tex
\documentclass[acmtkdd]{acmsmall} 

\usepackage[ruled]{algorithm2e}

\SetAlFnt{\small}
\SetAlCapFnt{\small}
\SetAlCapNameFnt{\small}
\SetAlCapHSkip{0pt}
\IncMargin{-\parindent}

\acmArticle{A}
\acmYear{YYYY}
\acmMonth{1}

\doi{0000001.0000001}

\issn{1234-56789}

\usepackage{amssymb}
\usepackage{latexsym}
\usepackage{color}
\usepackage{xspace}
\usepackage[noend]{algorithmic}
\usepackage{enumitem}
\usepackage{multirow}
\usepackage{graphicx}
\usepackage{subfigure}
\usepackage{url}

\usepackage{mathtools}

\graphicspath{{plots/}{setcplots/}{./}}

\begin{document}
	
	\markboth{S. Arifuzzaman et al.}{Distributed-Memory Parallel Algorithms for Counting and Listing Triangles in Big Graphs}
	
	\title{Distributed-Memory Parallel Algorithms for Counting and Listing Triangles in Big Graphs}
	\author{SHAIKH ARIFUZZAMAN
		\affil{University of New Orleans}
		MALEQ KHAN
		\affil{Texas A\&M University--Kingsville}
		MADHAV MARATHE
		\affil{Virginia Tech}
		}

	\begin{abstract}
		
		Big graphs (networks) arising in numerous application areas pose significant challenges for graph analysts as these graphs grow to billions of nodes and edges and are prohibitively large to fit in the main memory. Finding the number of triangles in a graph is an important problem in the mining and analysis of graphs. In this paper, we present two efficient MPI-based distributed memory parallel algorithms for counting triangles in big graphs. The first algorithm employs overlapping partitioning and efficient load balancing schemes to provide a very fast parallel algorithm. The algorithm scales well to networks with billions of nodes and can compute the exact number of triangles in a network with 10 billion edges in 16 minutes. The second algorithm divides the network into non-overlapping partitions leading to a space-efficient algorithm. Our results on both artificial and real-world networks demonstrate a significant space saving with this algorithm. We also present a novel approach that reduces communication cost drastically leading the algorithm to both a space- and runtime-efficient algorithm. Further, we demonstrate how our algorithms can be used to list all triangles in a graph and compute clustering coefficients of nodes. Our algorithm can also be adapted to a parallel approximation algorithm using an edge sparsification method. 
	\end{abstract}

\category{G.2.2}{Discrete Mathematics}{Graph Theory}[Graph Algorithms]
\category{D.1.3}{Programming Techniques}{Concurrent Programming}[Parallel Programming]
\category{H.2.8}{Database Management}{Database Applications}[Data Mining]

\terms{Algorithm, Experimentation, Performance}


\keywords{triangle-counting, clustering-coefficient, massive networks, parallel algorithms, social networks, graph mining.}

\begin{bottomstuff}
	
	This work has been partially supported by DTRA CNIMS Contract HDTRA1-11-D-0016-0001, DTRA
Grant HDTRA1-11-1-0016, DTRA NSF NetSE Grant CNS-1011769 and NSF SDCI Grant OCI-1032677.

	Some preliminary results of the work presented in this paper have appeared in the proceedings of CIKM 2013~\cite{Patric-triangle} and HPCC 2015~\cite{Space-triangle}.
	
	Author's addresses: Shaikh Arifuzzaman is with the University of New Orleans, New Orleans, LA 70148. Maleq Khan is with Texas A\&M University--Kingsville, TX 78363. Madhav Marathe is with Biocomplexity Institute and the Department of Computer Science, Virginia Tech, Blacksburg, VA, 24060.  E-mail: smarifuz@uno.edu, maleq.khan@tamuk.edu, mmarathe@vt.edu. Most of this work was done when all authors were with the Biocomplexity Institute of Virginia Tech.
\end{bottomstuff}

\maketitle

\newcommand{\todo}[1]{{\textcolor{red}{\textbf{To do:} #1}}}
\newcommand{\red}[1]{{\textcolor{red}{ #1}}}
\newcommand{\fsc}{}  
\newcommand{\afsc}{}   
\newcommand{\isc}{}

\input{introduction}

\input{preliminaries}

\input{serialalgo}

\input{analysis}

\input{parallelalgo-aop}

\input{performance-aop}

\input{background-anop}

\input{parallelalgo-anop}

\input{sparsification}

\input{trianglelisting}

\input{clusteringco}

\input{application}

\input{conclusion}

\bibliographystyle{ACM-Reference-Format-Journals}
\bibliography{CCpaper}  

\end{document}

%% file: introduction.tex
\section{Introduction}\label{sec:introduction}

Counting triangles in a graph is a fundamental and important algorithmic problem in graph analysis, and its solution can be used in solving many other problems such as the computation of clustering coefficient, transitivity, and triangular connectivity \cite{alon-network-motifs,chu11triangle}. Existence of triangles and the resulting high clustering coefficient in a social network reflect some common theories of social science, e.g., \textit{homophily} where people become friends with those similar to themselves and \textit{triadic closure} where people who have common friends tend to be friends themselves \cite{socialMiller}. Further, triangle counting has important
applications in graph mining such as detecting spamming activity and
assessing content quality in social networks \cite{BECC}, uncovering the thematic structure of the web \cite{eckmann2002theme}, query
planning optimization in databases \cite{YOSSEFF}, and detecting communities or clusters in social and information networks \cite{PratPerez2016PTT}.

Graph is a powerful abstraction for representing underlying relations in large unstructured datasets. Examples include  the web graph \cite{Broder-web}, various social networks \cite{twitter}, biological networks \cite{bioNet}, and many other information networks. In the era of big data, the emerging graph data is also very large. Social networks such as Facebook and Twitter have millions to billions of users \cite{chu11triangle,fb-stat}. Such big graphs motivate the need for efficient parallel algorithms. Furthermore, these massive graphs pose another challenge of a large memory
requirement. These graphs may not fit in the main memory of a single processing unit, and only a small part of the graph is available to a processor.

Counting triangles and related problems such as computing clustering coefficients have a rich history \cite{ALON,SCHANK,LATAPY,tsourakakis2009doulion,SURI,gpu_Green2014,park-cikm2,triangle_shun2015}. Despite the fairly large volume of work addressing this problem, only recently has attention been given to the problems associated with big graphs. Several techniques can be employed to deal with such graphs: streaming algorithms \cite{srikanta-cikm,BECC}, sparsification based algorithms \cite{tsourakakis2009doulion,Wu2016CT}, external-memory algorithms \cite{chu11triangle}, and parallel algorithms \cite{SURI,pinar-mapreduce,srikanta-cikm}. The streaming and sparsification based algorithms are approximation algorithms. Note that approximation algorithms provide an overall (global) estimate of the number of triangles in the graph, which might not be used to count triangles incident on individual nodes (local triangles) with reasonable accuracy. Thus certain local patterns such as local clustering coefficient distribution can not be computed with approximation algorithms. Exact algorithms are necessary to discover such local patterns. 

External memory algorithms can provide exact solution, however they can be very I/O intensive leading to a large runtime. Efficient parallel algorithms can solve the problem of a large runtime by distributing computing tasks to multiple processors. Over the last couple of years, several parallel algorithms, both shared memory and distributed memory (MapReduce or MPI) based, have been proposed.

A shared memory parallel algorithm is proposed in \cite{srikanta-cikm} for counting triangles in a streaming setting. The algorithm provides approximate counts. The paper reports scalability using only $12$ cores. Two other shared memory algorithms have been presented recently in \cite{triangle_shun2015,triangle_hasan2013}: the reported speedups with the first algorithm vary between $17$ and $50$ with $64$ cores. The second paper reports speedups using only $32$ cores, and the obtained speedups are due to both approximation and parallelization. Although these algorithms are useful, shared memory systems with a large number of processors and at the same time sufficiently large memory per processor are not widely available. Further, the overhead for locking and synchronization mechanism required for concurrent read and write access to shared data might restrict their scalability. A GPU-based parallel algorithm is proposed recently in \cite{gpu_Green2014} which achieves a speedup of only $32$ with $2880$ streaming processors.

There exist several algorithms based on the MapReduce framework. Suri et al. presented two algorithms for counting the exact number of triangles \cite{SURI}. The first algorithm generates huge volumes of intermediate data and requires a significantly large amount of time and memory. The second algorithm suffers from redundant counting of triangles. Two papers by Park et al. \cite{park-cikm,park-cikm2} achieved some improvement over the second algorithm of Suri et al., although the redundancy is not entirely eliminated. Another MapReduce based parallelization of a wedge-based sampling technique is proposed in \cite{pinar-mapreduce}, which is also an approximation algorithm.

MapReduce framework provides several advantages such as fault tolerance, abstraction of parallel
computing mechanisms, and ease of developing a quick prototype or program.  However, the overhead for doing so results in a larger runtime. On the other hand, MPI-based systems provide the advantages of defining and controlling parallelism from a granular level, implementing application specific optimizations such as load
balancing, memory and message optimization.

In this paper, we present fast algorithms for counting the \textit{exact} number of triangles. Our algorithms store a small portion of input graph in the main memory of each processor and can work on big graphs. 
Below are the summaries of our contributions.

\textbf{i. A fast parallel algorithm:} We propose an MPI based parallel algorithm that employs an overlapping partitioning scheme and a novel load balancing scheme. The overlapping partitions eliminate the need for message exchanges leading to a fast algorithm. The algorithm scales almost linearly with the number of processors, and is able to process a graph with 1 billion nodes and 10 billion edges in 16 minutes. To the best of our knowledge, this is the first MPI based parallel algorithm in literature for counting triangles in massive graphs.

\textbf{ii. A space efficient parallel algorithm:} We present a space-efficient MPI based parallel algorithm which divides the graph into non-overlapping partitions and achieves a significant space efficiency over the first algorithm. This algorithm requires inter-processor communications to count a certain type of triangles. However, we present a novel  approach that reduces communication cost drastically without requiring additional space, which leads to both a space- and runtime-efficient algorithm. Our adaptation of a parallel partitioning scheme by computing a novel cost function offers additional runtime efficiency to the algorithm.

\textbf{iii. Sequential algorithm and node ordering:} We show, both theoretically and experimentally, a simple modification of a state-of-the-art sequential algorithm for counting triangles improves its performance and use this modified algorithm in the development of our parallel algorithm. We also present a proof of the optimal node ordering that minimizes the computational cost of this sequential algorithm.

\textbf{iv. Parallel computation of clustering coefficients:} In a sequential setting, an algorithm for counting triangles can be directly used for computing clustering coefficients of the nodes by simply keeping the counts of triangles for each node individually. However, in a distributed-memory parallel system, combining the counts from all processors for all nodes poses another level of difficulty. We show how our algorithm for triangle counting can be used to compute clustering coefficients in parallel.

\textbf{v. Parallel approximation using sparsification technique:} Although we present algorithms for counting the exact number of triangles in massive graphs, our algorithm can be used for approximate counting in conjunction with an edge sparsification technique \cite{tsourakakis2009doulion}. We show how this technique can be adapted to our parallel algorithms and that our parallel sparsification improves the accuracy of the approximation over the sequential sparsification \cite{tsourakakis2009doulion}.

\textbf{Organization.} The rest of the paper is organized as follows. The preliminary concepts, notations and datasets are briefly described in Section \ref{sec:prelim}. We discuss sequential
algorithms for counting triangles and present a proof for the optimal node ordering in Section \ref{sec:seqalg} and \ref{sec:optimal_ordering}, respectively. Our parallel algorithms for counting triangles are presented in Section \ref{sec:patric} and \ref{sec:space}. The parallelization of the sparsification technique is given in Section \ref{sec:sparse}. We show in Section \ref{sec:trianglelisting} how we can list all triangles in graphs in parallel. Section \ref{sec:clustering} presents a parallel algorithm for computing clustering coefficients of nodes. We discuss some applications of counting triangles in Section \ref{sec:applications} and conclude in Section \ref{sec:conclusion}.

%% file: preliminaries.tex
\section{Preliminaries}
\label{sec:prelim}

The given graph is denoted by $G(V,E)$, where $V$ and $E$ are the
sets of nodes (vertices) and edges, respectively, with $m = |E|$ edges
and $n = |V|$ nodes labeled as $0, 1, 2, \dots, n-1$. We assume the graph $G(V,E)$ is undirected. If $(u,v)\in E$, we say $u$ and $v$ are neighbors of each other. The set of all neighbors of $v \in V$ is denoted by $\mathcal{N}_v$, i.e., $\mathcal{N}_v=\{u \in V | (u,v) \in E\}$. The degree of $v$ is $d_v = |\mathcal{N}_v|$.

A triangle in $G$ is a set of three nodes $u,v,w \in V$ such that there is an
edge between each pair of these three nodes, i.e., $(u,v), (v,w),
(w,u) \in E$. The number of triangles containing node $v$ (in other
words, triangles incident on $v$) is denoted by $T_v$. Notice that the
number of triangles containing node $v$ is the same as the number of
edges among the neighbors of $v$, i.e., $T_v = |\left\{(u,w) \in E : u, w \in \mathcal{N}_v \right\}|$.

The clustering coefficient (CC) of a node $v \in V$, denoted by $C_v$ is the ratio of the number of edges between neighbors of $v$ to the number of all possible edges between neighbors of $v$. Then, we have
\begin{eqnarray*}
C_v = \frac{T_v}{{d_v \choose 2}} = \frac{2T_v}{d_v(d_v - 1)}.
\end{eqnarray*}

Let $p$ be the number of processors used in the computation, which we denote by $P_0,P_1,\dots, P_{p-1}$ where each subscript refers to the rank of a processor.




\textbf{Datasets.} 
We use both real world and artificially generated networks for our experiments. A summary of all the networks is provided in Table \ref{table:dataset}. Miami \cite{miamiRef} is a synthetic, but realistic, social contact network for Miami city. Twitter, LiveJournal, Email-Enron, web-BerkStan, and web-Google are real-world networks. Artificial network  PA($n,d$) is generated using the preferential attachment (PA) model \cite{barabasi99} with $n$ nodes and average degree $d$. Network Gnp($n,d$) is generated using the Erd\H{o}s-R\'{e}yni random graph model \cite{bollobas01Erdos}, also known as $G(n,q)$ model, with $n$ nodes and edge probability $q = \frac{d}{n-1}$ so that the expected degree of each node is $d$. Both real-world and PA($n,d$) networks have very skewed degree distributions. Networks having such distributions create difficulty in partitioning and balancing loads and thus give us a chance to measure the performance of our algorithms in some of the worst case scenarios.

\begin{table}
	\tbl{Dataset used in our experiments. K, M and B denote thousands, millions and billions, respectively.\label{table:dataset}}{
	\begin{tabular}{ | l | l | l | l |}
		\hline
		{\bf Network} & {\bf Nodes} & {\bf Edges} & {\bf Source} \\ \hline
		Email-Enron   &  $37$K    &   $0.36$M     & SNAP \cite{SNAP}\\ \hline
		web-Google   &  $0.88$M    &   $5.1$M     & SNAP \cite{SNAP}\\ \hline
		web-BerkStan  &  $0.69$M    &   $6.5$M   & SNAP \cite{SNAP}\\ \hline
		Miami         &  $2.1$M    &   $50$M     & \cite{miamiRef}\\ \hline
		LiveJournal   &  $4.8$M    &  $43$M     & SNAP \cite{SNAP}\\ \hline
		Twitter       &  $42$M   &     $2.4$B  & \cite{twitterData}\\ \hline
		Gnp($n,d$)    &  $n$      &   $\frac{1}{2}nd$     & Erd\H{o}s-R\'{e}yni \cite{bollobas01Erdos}\\ \hline
		PA($n,d$)     &  $n$     &     $\frac{1}{2}nd$    & Pref. Attachment \cite{barabasi99}\\ \hline
	\end{tabular}}
\end{table}

\textbf{Computation Model.} We develop parallel algorithms for message
passing interface (MPI) based distributed-memory parallel systems,
where each processor has its own local memory. The processors do not
have any shared memory, one processor cannot directly access the local
memory of another processor, and the processors communicate via
exchanging messages using MPI.

\textbf{Experimental Setup.} We perform our experiments using a high performance computing cluster with 64 computing nodes (QDR InfiniBand interconnect), 16 processors (Sandy Bridge E5-2670, 2.6GHz) per node, memory 4GB/processor, and operating system CentOS Linux 6.

%% file: serialalgo.tex
\section{Sequential Algorithms}
\label{sec:seqalg}

In this section, we discuss sequential algorithms for counting triangles and show that a simple modification to a state-of-the-art algorithm improves both runtime and space requirement. Although the modification seems quite simple, and others might have used it previously, to the best of our knowledge, our analysis is the first to show that such modification improves the performance significantly. Our parallel algorithms are based on this improved algorithm. 

A simple but efficient algorithm \cite{SURI,SCHANK} for counting triangles is: for each node $v\in V$, find the number of edges among its neighbors, i.e., the number of pairs of neighbors that complete a triangle with vertex $v$. In this method, each triangle $(u,v,w)$ is counted six times. Many existing algorithms \cite{Schank2005Triangle,LATAPY,SCHANK,chu11triangle,SURI} provide significant improvement over the above method. A very comprehensive survey of the sequential algorithms can be found in \cite{LATAPY,SCHANK}. One of the state of the art algorithms, known as NodeIterator++, as identified in two recent papers \cite{chu11triangle,SURI}, is shown in Figure \ref{algo:2}. Both \cite{chu11triangle} and \cite{SURI} use this algorithm as a basis of their external-memory and parallel algorithm, respectively.

\begin{figure}[!t]
\begin{center}
\fbox{
\begin{minipage}[c] {0.8\linewidth}
\begin{algorithmic}[1]
\STATE $T \leftarrow 0$ \hspace{0.2in} \{$T$ stores the count of triangles\}
\FOR {$v \in V$}
	\FOR {$u \in \mathcal{N}_v$ and $v \prec u$}
		\FOR {$w \in \mathcal{N}_v$ and $u \prec w$}
			\IF {$(u,w) \in E$}
				\STATE $T \leftarrow T+1$
			\ENDIF	
		\ENDFOR
	\ENDFOR
\ENDFOR
\end{algorithmic}
\end{minipage}
}\fsc
\end{center}
\fsc 
\caption{Algorithm NodeIterator++, where $\prec$ is the degree based ordering of the nodes defined in Equation \ref{eqn:order}.}
\label{algo:2}
\afsc 
\end{figure}

The algorithm NodeIterator++ uses a total ordering $\prec$ of the nodes to avoid duplicate counts of the same triangle. Any arbitrary ordering of the nodes, e.g., ordering the nodes based on their IDs, makes sure each triangle is counted exactly once -- counts only one among the six possible permutations. However, NodeIterator++ incorporates an interesting node ordering based on the degrees of the nodes, with ties broken by node IDs, as defined below:
\begin{equation} \label{eqn:order}
u \prec v \iff d_u < d_v \mbox{ or } (d_u = d_v \mbox{ and } u < v).
\end{equation}

\begin{definition} [effective degree] \label{dfn:eff_deg}
While $\mathcal{N}_v$ is the set of all neighbors of $v \in V$, let $N_v=\{u\in V | (u,v) \in E \land v \prec u\}$, i.e., $N_v$ is the set of neighbors $u$ of $v$ such that $v \prec u$. We define $\hat{d}_v = |N_v|$ as the effective degree of $v$.	
\end{definition}

The degree based ordering can improve the running time. Assuming $\mathcal{N}_v$, for all $v$, are sorted and a binary search is used to check $(u,w) \in E$, a runtime of $O\left(\sum_v{(\hat{d}_v d_v + \hat{d}_v^2 \log d_{\max})}\right)$ can be shown, where $d_{\max} = \max_v{d_v}$. This runtime is minimized when $\hat{d}_v$ values of the nodes are as close to each other as possible, although, for any ordering of the nodes, $\sum_v{\hat{d}_v} = m$ is invariant.

Notice that in the degree-based ordering, variance of the $\hat{d}_v$ values are reduced significantly. We also observe that for the same reason, degree-based ordering of the nodes helps keep the loads among the processors balanced, to some extent, in a parallel algorithm as discussed in detail in Section \ref{sec:patric}.

\begin{figure}
\begin{center}
\fbox{
\begin{minipage}[c] {0.8\linewidth}
\begin{algorithmic}[1]
\STATE \{Preprocessing: Line 2-6\}
\FOR {each edge $(u,v)$}
    \STATE if {$u \prec v$}, store $v$ in $N_u$
    \STATE else store $u$ in $N_v$
\ENDFOR
\FOR {$v \in V$}
	\STATE sort $N_v$ in ascending order
\ENDFOR
\STATE $T \leftarrow 0$ \hspace{0.2in} \{$T$ is the count of triangles\}
\FOR {$v \in V$}
	\FOR {$u \in N_v$}
		\STATE $S \leftarrow N_v \cap N_u$
		\STATE $T \leftarrow T+|S|$
	\ENDFOR
\ENDFOR
\end{algorithmic}
\end{minipage}
}
\end{center}
\fsc
\caption{Algorithm NodeIteratorN, a modification of NodeIterator++.}
\label{algo:nodeiteratorn}
\afsc
\end{figure}

A simple modification of NodeIterator++ is as follows: perform comparison $u \prec v$ for each edge $(u,v) \in E$ in a preprocessing step rather than doing it while counting the triangles. This preprocessing step reduces the total number of $\prec$ comparisons to $O(m)$ from $\sum_v{\hat{d}_v d_v}$ and allows us to use an efficient set intersection operation. For each edge $(v,u)$, $u$ is stored in $N_v$ if and only if $ v \prec u$. The modified algorithm NodeIteratorN is presented in Figure \ref{algo:nodeiteratorn}. All triangles containing node $v$ and any $u \in N_v$ can be found by set intersection $N_u \cap N_v$ (Line 10 in Figure \ref{algo:nodeiteratorn}). The correctness of NodeIteratorN is proven in Theorem \ref{thm:correct}.

\begin{theorem} \label{thm:correct}
Algorithm NodeIteratorN counts each triangle in $G$ once and only once.
\end{theorem}
\begin{proof}
Consider a triangle $(x_1, x_2, x_3)$ in $G$, and without the loss of generality, assume that $x_1 \prec x_2 \prec x_3$. By the construction of $N_x$ in the preprocessing step, we have $x_2, x_3 \in {N_{x_1}}$ and $x_3 \in N_{x_2}$. When the loops in Line 8-9 begin with $v = x_1$ and $u = x_2$, node $x_3$ appears in $S$ (Line 10-11), and the triangle $(x_1, x_2, x_3)$ is counted once. But this triangle cannot be counted for any other values of $v$ and $u$ because $x_1 \notin N_{x_2}$ and $x_1,x_2 \notin N_{x_3}$.
\end{proof}

In NodeIteratorN, when $N_v$ and $N_u$ are sorted, $N_u \cap N_v$ can be computed in $O(\hat{d}_u + \hat{d}_v)$ time. Then we have $O\left(\sum_{v\in V}{d_v\hat{d}_v}\right)$ time complexity for NodeIteratorN as shown in Theorem \ref{thm:time}, in contrast to $O\left(\sum_v{(\hat{d}_v d_v + \hat{d}_v^2 \log d_{\max})}\right)$ for NodeIterator++.

\begin{theorem} \label{thm:time}
The time complexity of algorithm NodeIteratorN is $O\left(\sum_{v\in V}{d_v\hat{d}_v}\right)$.
\end{theorem}
\begin{proof}
 Time for the construction of $N_v$ for all $v$ is $O\left(\sum_{v}{d_v}\right)$ $= O(m)$, and sorting these $N_v$ requires $O\left(\sum_{v}{\hat{d}_v\log \hat{d}_v}\right)$ time. Now, computing intersection $N_v \cap N_u$ takes $O(\hat{d}_u + \hat{d}_v)$ time. Thus, the time complexity of NodeIteratorN is
\begin{eqnarray*}
& & O(m) + O\left(\sum_{v \in V}{\hat{d}_v\log \hat{d}_v}\right) + O\left(\sum_{v\in V}\sum_{u\in N_v}{(\hat{d}_u + \hat{d}_v)}\right) \\
& = & O\left(\sum_{v \in V}{\hat{d}_v\log \hat{d}_v}\right) + O\left(\sum_{(v,u)\in E}{(\hat{d}_u + \hat{d}_v)}\right) \\
& = & O\left(\sum_{v \in V}{\hat{d}_v\log \hat{d}_v}\right) + O\left(\sum_{v\in V}{d_v\hat{d}_v}\right) = O\left(\sum_{v\in V}{d_v\hat{d}_v}\right).
\end{eqnarray*}

The second last step follows from the fact that for each $v\in V$, term $\hat{d}_v$ appears $d_v$ times in this expression.
\end{proof}

Notice that the set intersection operation can also be used with NodeIterator++ by replacing Line 4-6 of NodeIterator++ in Figure \ref{algo:2} with the following three lines as shown in \cite{chu11triangle} (Page 674):

\begin{center}
\fbox{
\begin{minipage}[c] {0.5\linewidth}
\begin{algorithmic}[1]
\STATE $S \leftarrow \mathcal{N}_v \cap \mathcal{N}_u$
\FOR {$w \in S$ and $u \prec w$}
	\STATE $T \leftarrow T+1$
\ENDFOR
\end{algorithmic}
\end{minipage}
}
\end{center}

However, with this set intersection operation, the runtime of NodeIterator++ is $O\left(\sum_{v}{d_v^2}\right)$ since $|\mathcal{N}_v| = d_v$, and computing $\mathcal{N}_v \cap \mathcal{N}_u$ takes $O(d_u + d_v)$ time. Further, the memory requirement for NodeIteratorN is half of that for NodeIterator++. NodeIteratorN stores $\sum_{v}{\hat{d}_v} = m$ elements in all $N_v$ and
NodeIterator++ stores $\sum_{v}{d_v} = 2m$ elements. Here we would like to note that the two algorithms presented in \cite{Schank2005Triangle,LATAPY} take the same asymptotic time complexity as NodeIteratorN. However, the algorithm in \cite{Schank2005Triangle} requires three times more memory than NodeIteratorN. The algorithm in \cite{LATAPY} requires more than twice the memory as NodeIteratorN, maintains a list of indices for all nodes, and the hidden constant in the runtime can be much larger. Our experimental results show that NodeIteratorN is significantly faster than NodeIterator++ for both real-world and artificial networks as presented in Table \ref{table:compareNI+andNIN}.

\begin{table}
\tbl{Running time for sequential algorithms
\label{table:compareNI+andNIN}}
{
\begin{tabular}{|l|c|c|c|} \hline
\multirow{2}{*}{ \textbf{Networks}} &  \multicolumn{2}{|c|}{\textbf{Runtime (sec.)}} & \multirow{2}{*}{\textbf{Triangles}} \\
\cline{2-3}
 & \footnotesize{NodeIterator++} & \footnotesize{NodeIteratorN} &  \\ \hline
Email-Enron     &    0.14 &  0.07 &      0.7M  	 \\ \hline
web-BerkStan    &     3.5 &  1.4   &  64.7M \\ \hline
LiveJournal     &    106  &  42   & 285.7M \\ \hline
Miami           &     46.35  &  32.3   &  332M \\ \hline
PA(25M, 50)       &  690    & 360  & 1.3M  \\ \hline
\end{tabular}
}
\end{table}

%% file: analysis.tex
\section{An Optimal Node Ordering} \label{sec:optimal_ordering}

A total ordering $\prec$ of the nodes helps avoid duplicate counts of the same triangle. Any ordering of the nodes, e.g., ordering based on node IDs, random ordering, $k$-coreness based ordering, make sure each triangle is counted exactly once. By avoiding duplicate counts, these orderings also improve running time of the algorithm. However, different orderings lead to different runtimes. Figure \ref{fig:order_comp} shows the runtime of our sequential algorithm for triangle counting with four orderings of nodes: ordering based on node IDs, degree, $k$-coreness, and random ordering. Node IDs and degrees are readily available with network data and do not require any additional computation. On the other hand, $k$-coreness based ordering requires computing coreness of nodes, and for random ordering, we generate $n$ random numbers. Figure \ref{fig:order_tc} shows the comparison of runtime of counting triangles without considering the cost for computing orderings. Figure \ref{fig:order_total} shows the comparison with total runtime of counting triangles and computing orderings. In both cases, degree based ordering provides the best runtime efficiency among all orderings. For networks with relatively even degree distribution such as Miami, all the orderings provide similar runtimes. However, for networks with skewed degree distribution, degree based ordering provides the least runtime. In our datasets, nodes with large degrees somehow appear at the beginning (having smaller IDs) giving ID based ordering almost the opposite effect of degree based ordering. As a result, ID based ordering provides the largest runtime for our datasets. 

\begin{figure*}[ht!]
	\centering
	\subfigure[Runtime for triangle counting without considering the runtime for computing ordering of nodes]{
		\includegraphics[width=0.47\textwidth]{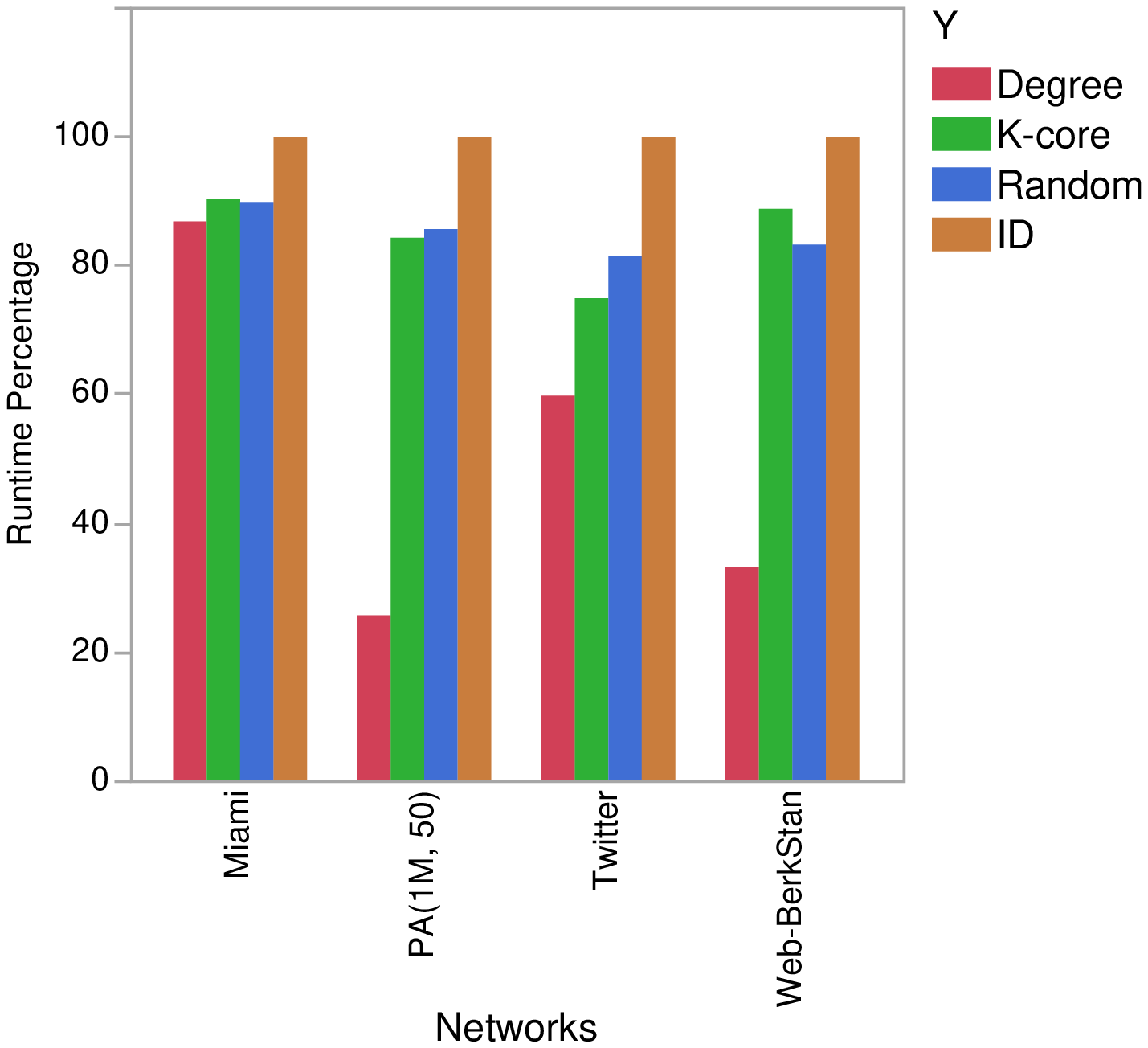}
		\label{fig:order_tc}
	}
	\subfigure[Total runtime for counting triangles and computing ordering of nodes]{
		\includegraphics[width=0.47\textwidth]{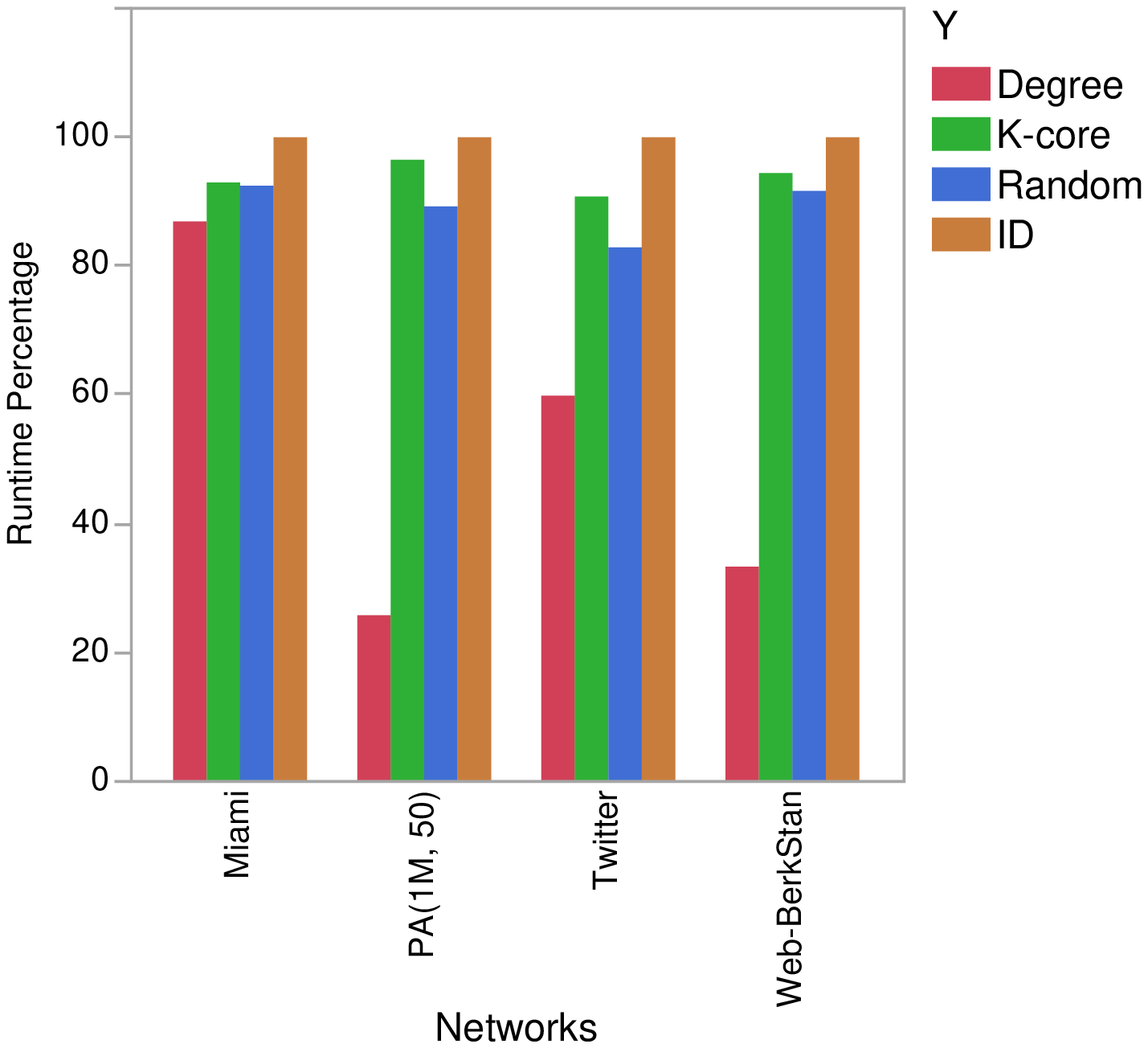}
		\label{fig:order_total}
	}
	\fsc
	\caption{Comparison of runtime of sequential triangle counting (NodeIteratorN) with four distinct orderings of nodes. For each network, we compute the percentage of runtime with respect to the maximum runtime given by any of these orderings. In all cases, the degree based ordering gives the least runtime. Note that we compute the average runtime from $25$ independent runs for the random ordering.}
	\label{fig:order_comp}
	\afsc
\end{figure*}

Now that our experimental results show degree based ordering provides the best runtime efficiency, next we show in Theorem \ref{thm:ordering} that the degree based ordering is, in fact, the optimal ordering that minimizes the runtime of algorithm \textit{NodeIteratorN}.

We denote the degree based ordering as $\prec_\mathcal{D}$ which is defined as follows: 
\begin{equation} \label{eqn:order2}
u {\prec}_\mathcal{D} v \iff d_u < d_v \mbox{ or } (d_u = d_v \mbox{ and } u < v).
\end{equation}
Assume there is another total ordering $\prec_\mathcal{K}$ based on some quantity $k_v$ of nodes $v$:
\begin{equation} \label{eqn:korder}
u {\prec}_\mathcal{K} v \iff k_u < k_v \mbox{ or } (k_u = k_v \mbox{ and } u < v).
\end{equation}
We now define a function which quantifies how ordering $\prec_\mathcal{K}$ agrees with $\prec_\mathcal{D}$ on the relative order of $x,y \in V$.


\begin{definition}[Agreement function Y] \label{dfn:agreement}
The \textit{agreement} function $Y: V \times V \rightarrow \mathbb{Z}$ is defined as follows:

\[ Y(x,y) = \left\{ 
  \begin{array}{l l}
    -1, & \quad \text{if $(x,y) \in E$ and $x {\prec}_{\mathcal{D}} y$ and $y \prec_\mathcal{K} x$ }\\
    1, & \quad \text{if $(x,y) \in E$ and $y {\prec}_{\mathcal{D}} x$ and $x {\prec}_{\mathcal{K}} y$}\\
    0, & \quad \text{Otherwise}
  \end{array} \right.\]
\end{definition}
  
It is, then, easy to see that $Y(x,y) = - Y(y,x)$.\\

We now prove an important result in the following lemma, which we subsequently use in Theorem \ref{thm:ordering}. 

\begin{lemma} \label{lemma:cxy}
For any $(x,y) \in E$, $Y(x,y)( d_x - d_y) \geq 0$.
\end{lemma}
\begin{proof}
Let $c_{xy} = Y(x,y)( d_x - d_y)$. If orderings $\prec_\mathcal{K}$ and $\prec_\mathcal{D}$ agree on the relative order of $x$ and $y$, then $Y(x,y)=0$ by definition, and hence, $c_{xy} = 0$. Otherwise, consider the following three cases. 
\begin{itemize}
\item $d_x = d_y$:
This gives $d_x-d_y=0$, and thus, $c_{xy} = 0$.
\item $d_x < d_y$:
We have $x \prec_\mathcal{D} y$ and $y \prec_\mathcal{K} x$, and thus, $Y(x,y) = -1$. Since $d_x - d_y < 0$, $c_{xy} > 0$.
\item $d_x > d_y$:
We have $y \prec_\mathcal{D} x$ and $x \prec_\mathcal{K} y$, and thus, $Y(x,y) = 1$. Since $d_x - d_y > 0$, $c_{xy} > 0$.
\end{itemize} 

Therefore, for any $(x,y) \in E$, $c_{xy}= Y(x,y)( d_x - d_y) \geq 0$.
\end{proof}

\begin{theorem} \label{thm:ordering}
Degree based ordering $\prec_\mathcal{D}$ minimizes the runtime for counting triangles using algorithm \textit{NodeIteratorN}.
\end{theorem}

\begin{proof}
Let  $\hat{d_v}$ be the effective degree of vertex $v$ with ordering $\prec_\mathcal{D}$. Then, the corresponding runtime for counting triangles is $\Theta \left( \sum_{i\in V}{d_i\hat{d}_i} \right)$. We provide a proof by contradiction. Assume that $\prec_\mathcal{D}$ is not an optimal ordering. Then there exists another ordering $\prec_\mathcal{K}$ that leads to a lower runtime for counting triangles than that of $\prec_\mathcal{D}$. Let $\prec_\mathcal{K}$ yields an effective degree $\tilde{d}$, the corresponding runtime for counting triangles is $ \Theta \left( \sum_{i\in V}{d_i\tilde{d}_i}\right)$. Let $C_{\mathcal{D}} = \sum_{i\in V}{d_i\hat{d}_i}$ and $C_{\mathcal{K}}=  \sum_{i\in V}{d_i\tilde{d}_i}$. Then, we have $C_{\mathcal{K}} <  C_{\mathcal{D}}$. 

Now, using Definition \ref{dfn:agreement}, the effective degree $\tilde{d_x}$ of node $x$ obtained by $\prec_\mathcal{K}$ can be expressed as,
\begin{eqnarray*}
\tilde{d}_x = \hat{d}_x + \sum_{y \in \mathcal{N}_x} {Y(x,y)}.
\end{eqnarray*}

Now, we have,

\begin{eqnarray*}
C_{\mathcal{K}}
&=& \sum_{x\in V}{d_x\tilde{d}_x} \\
&=& \sum_{x\in V}{d_x \left( \hat{d}_x + \sum_{y \in \mathcal{N}_x} {Y(x,y)}\right) } \\
&=& \sum_{x\in V} d_x \hat{d}_x + \sum_{x\in V} \left( d_x \sum_{y \in \mathcal{N}_x} {Y(x,y)}\right) \\
&=& \sum_{x\in V} d_x \hat{d}_x + \sum_{(x,y) \in E} \left( d_x Y(x,y) + d_y Y(y,x) \right) \\ 
&=& \sum_{x\in V} d_x \hat{d}_x + \sum_{(x,y) \in E} Y(x,y) \left( d_x - d_y \right).
\end{eqnarray*}

The second last step follows from rearranging terms of the second summation and distributing them over edges. The last step follows from the fact that $Y(y,x) = -Y(x,y)$. Now, from Lemma \ref{lemma:cxy} we have, $Y(x,y)( d_x - d_y) \geq 0$ for any $(x,y) \in E$.
Thus, $\sum_{(x,y) \in E} Y(x,y) \left( d_x - d_y \right) \geq 0$, and therefore, 
\begin{eqnarray*}
C_{\mathcal{K}} \geq \sum_{x\in V} d_x \hat{d}_x = C_{\mathcal{D}}.
\end{eqnarray*}

This contradicts our assumption of $C_{\mathcal{K}} < C_{\mathcal{D}}$.
Therefore, degree based ordering $\prec_\mathcal{D}$ is an optimal ordering which minimizes the runtime for counting  triangles of our algorithm.  
\end{proof}

We use algorithm \textit{NodeIteratorN} with degree based ordering in our parallel algorithms.

%% file: parallelalgo-aop.tex
\section{A Fast Parallel Algorithm with Overlapping Partitioning}
\label{sec:patric}

In this section, we present our fast parallel algorithm for counting triangles in massive graphs with overlapping partitioning and novel load balancing schemes.

\subsection{Overview of the Algorithm}
\label{sec:ParallelOverview}

We assume that the graph is massive and does not fit in the local memory of a single computing node. Only a part of the entire graph is available to a processor. Let $p$ be the number of processors used in the computation. The graph is partitioned into $p$ partitions, and each processor $P_i$ is assigned one such partition $G_i(V_i, E_i)$ (formally defined below). $P_i$ performs computation on its partition $G_i$. The main steps of our fast parallel algorithm are given in Figure \ref{algo:patric}. In the following subsections, we describe the details of these steps and several load balancing schemes.

\begin{figure}[!ht]
\begin{center}
\fbox{
\begin{minipage}[c] {0.9\linewidth}
\begin{algorithmic}[1]
\STATE{Each processor $P_i$, in parallel, executes the following:(lines 2-4)}
    \STATE $G_i(V_i,E_i) \leftarrow \textsc{ComputePartition}(G, i)$
    \STATE $T_i \leftarrow \textsc{CountTriangles}(G_i, i)$
\STATE $\textsc{Barrier}$
\STATE Find $T=\sum_i{T_i}$ 
\RETURN  $T$
\end{algorithmic}
\end{minipage}
}
\end{center}
\fsc
\caption{The main steps of our fast parallel algorithm.}
\label{algo:patric}
\afsc
\end{figure}

\begin{figure*}[!th]
\hfill
\begin{minipage}[t]{.32\textwidth}
\begin{center}
\centerline{\includegraphics[width=1.0\textwidth]{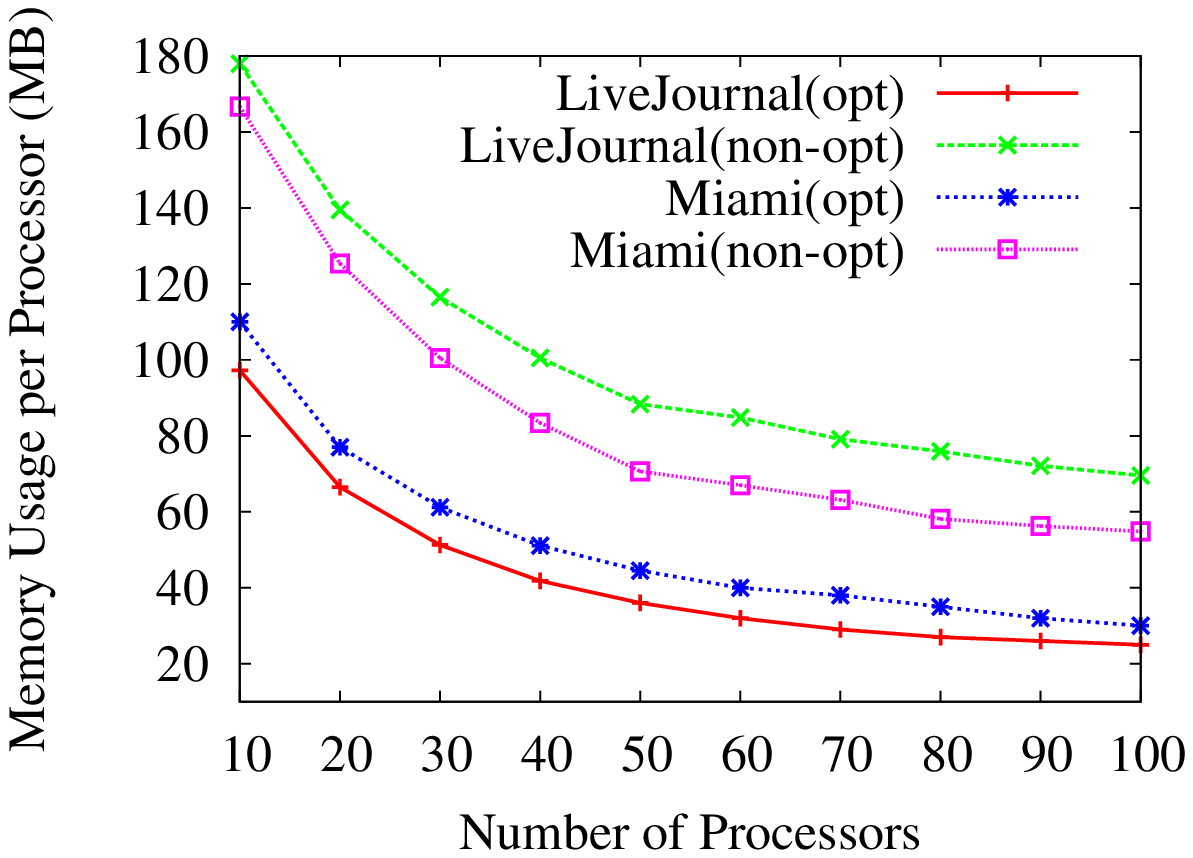}}
\caption{Memory usage with optimized and non-optimized data storing.}
\label{fig:memory}
\end{center}
\end{minipage}
\hfill
\begin{minipage}[t]{.32\textwidth}
\begin{center}
\centerline{
\raisebox{0.63in}{
\fbox{
\begin{minipage}[c] {0.90\linewidth}
	\scriptsize
\begin{algorithmic}[1]
\FOR {$v \in V_{i}$}
	\STATE sort $N_v$ in ascending order
\ENDFOR
\STATE $T \leftarrow 0$
\FOR {$v \in V_{i}^c$}
	\FOR {$u \in N_v$}
		\STATE $S \leftarrow N_v \cap N_u$
		\STATE $T \leftarrow T+|S|$
	\ENDFOR
\ENDFOR
\RETURN  $T$	
\end{algorithmic}
\end{minipage}
}
}
}
\caption{Algorithm executed by processor $P_i$ to count triangles in $G_i(V_i, E_i)$.}
\label{algo:tcount}
\end{center}
\end{minipage}
\hfill
\begin{minipage}[t]{.32\textwidth}
\begin{center}
\centerline{\scalebox{0.40}{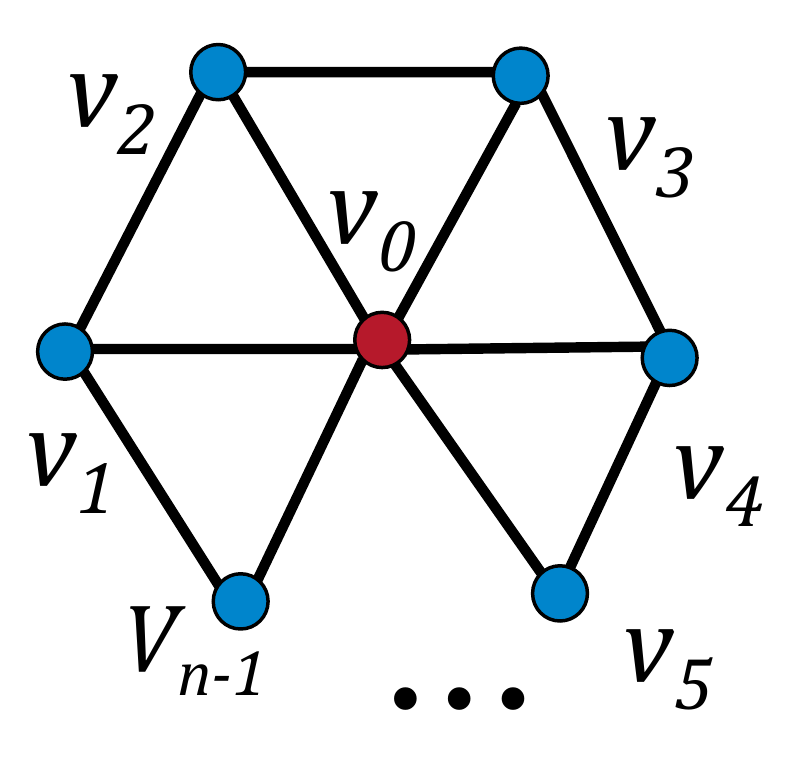}}
\caption{A network with a skewed degree distribution: $d_{v_0} = n-1$, $d_{v_{i\neq0}} = 3$.}
\label{fig:wheel}
\end{center}
\end{minipage}
\hfill
\end{figure*}

\subsection{Partitioning the Graph}
\label{sec:partition}

The memory restriction poses a difficulty where the graph must be partitioned in such a way that the memory required to store a partition is minimized and at the same time the partition contains sufficient information to minimize communications among the processors.
For the input graph $G(V,E)$, processor $P_i$ works on $G_i(V_i, E_i)$, which is a subgraph of $G$ induced by $V_i$. The subgraph $G_i$ is constructed as follows: First, set of nodes $V$ is partitioned into $p$ disjoint subsets $V_0^c, V_1^c, \dots, V_{p-1}^c$, such that, for any $j$ and $k$, $V_j^c \cap V_k^c = \emptyset$ and $\bigcup_{k} V_k^c = V$. Second, set $V_i$ is constructed containing all nodes in $V_i^c$ and $\bigcup_{v \in V_i^c} N_v$.  Edge set $E_i \subset E$ is the set of edges $\{(u,v): u \in V_i \text{\ and\ } v \in N_u \}$.

Each processor $P_i$ is responsible for counting triangles incident on the nodes in $V_i^c$. We call any node $v \in V_{i}^c$ a $core$ node of partition $i$. Each $v\in V$ is a core node in exactly one partition. How the nodes in $V$ are distributed among the core sets $V_i^c$ for all $P_i$ affect the load balancing and hence performance of the algorithm crucially. Later in Section \ref{sec:load}, we present several load balancing schemes and the details of how sets $V_i^c$ are constructed.

Now, $P_i$ stores the set of neighbors $N_v$ of all $v \in V_i$. Notice that for a node $w \in (V_i - V_i^c)$, neighbor set $N_w$ may contain some nodes $x \notin V_i$. Such nodes $x$ can be safely removed from $N_w$ and the number of triangles incident on all $v \in V_i^c$ can still be computed correctly. But, the presence of these nodes in $N_w$ does not affect the correctness of the algorithm either. However, as our experimental results in Figure \ref{fig:memory} show, we can save about 50\% of memory space by not storing such nodes $x\notin V_i$ in $N_w$. Figure \ref{fig:memory} also demonstrates the memory-scalability of our algorithm: as the more processors are used, each processor consumes less memory space.

\subsection{Counting Triangles}

Once each processor $P_i$ has its partition $G_i(V_i, E_i)$, it uses the improved sequential algorithm \emph{NodeIteratorN} presented in Section \ref{sec:seqalg} to count triangles in $G_i$ for each \emph{core} node $v \in V_i^c$. Neighbor sets $N_w$ for the nodes $w \in V_i - V_i^c$ only help in finding the edges among the neighbors of the core nodes.

To be able to use an efficient intersection operation, $N_v$ for all $v \in V_i$ are sorted. The code executed by $P_i$ is given in Figure \ref{algo:tcount}. Once all processors complete their counting steps, the counts from all processors are aggregated into a single count by an MPI reduce function, which takes $O(\log p)$ time. Ordering of the nodes, construction of $N_v$, and disjoint node partitions $V_i^c$ make sure that each triangle in the network appears exactly in one partition $G_i$. Thus, the correctness of the sequential algorithm \emph{NodeIteratorN} shown in Section \ref{sec:seqalg} ensures that each triangle is counted exactly once.

\subsection{Load Balancing}
\label{sec:load}

To reduce the runtime of a parallel algorithm, it is desirable that no processor remains idle and all processors complete their executions almost at the same time. In Section \ref{sec:seqalg}, we discussed how degree based ordering of the nodes can reduce the runtime of the sequential algorithm, and hence it reduces the runtime of the local computation in each processor $P_i$. We observe that, interestingly, this ordering also provides load balancing to some extent, both in terms of runtime and space, at no additional cost. Consider the example network shown in Figure \ref{fig:wheel}. With an arbitrary ordering of the nodes, $|N_{v_0}|$ can be as much as $n-1$, and a single processor which contains $v_0$ as a core node is responsible for counting all triangles incident on $v_0$. Then the runtime of the parallel algorithm can essentially be same as that of a sequential algorithm. With the degree-based ordering, we have $|N_{v_0}| = 0$ and $|N_{v_i}| \le 3$ for all $i$. Now if the core nodes are equally distributed among the processors, both space and computation time are almost balanced.

Although degree-based ordering helps mitigate the effect of skewness in degree distribution and balance load to some extent, working with more complex networks and highly skewed degree distribution reveals that distributing core nodes equally among the processors does not make the load well-balanced in many cases. Figure \ref{fig:speedup_initial} shows speedup of the parallel algorithm with an equal number of core nodes assigned to each processor. LiveJournal network shows poor speedup, whereas the Miami network shows a relatively better speedup. This poor speedup for LiveJournal network is a consequence of highly imbalanced computation load across the processors as shown in Figure \ref{fig:ljmiami_rankwise}. Unlike Miami network, LiveJournal network has a very skewed degree distribution. (Note that we used 100 processors for our experiments on load distribution. Although we could use a higher number of processors, using fewer processors helped demonstrate the pattern of imbalance of loads more clearly. In our subsequent experiments on scalability, we use a higher number of processors. In fact, we show that our algorithm scales to a larger number of processors when networks grow larger.)

In the next section, we present several load balancing schemes that improve the performance of our algorithm significantly.

\begin{figure*}[!tbh]
\hfill
\begin{minipage}[t]{.32\textwidth}
\begin{center}
\centerline{\includegraphics[width=1.0\textwidth]{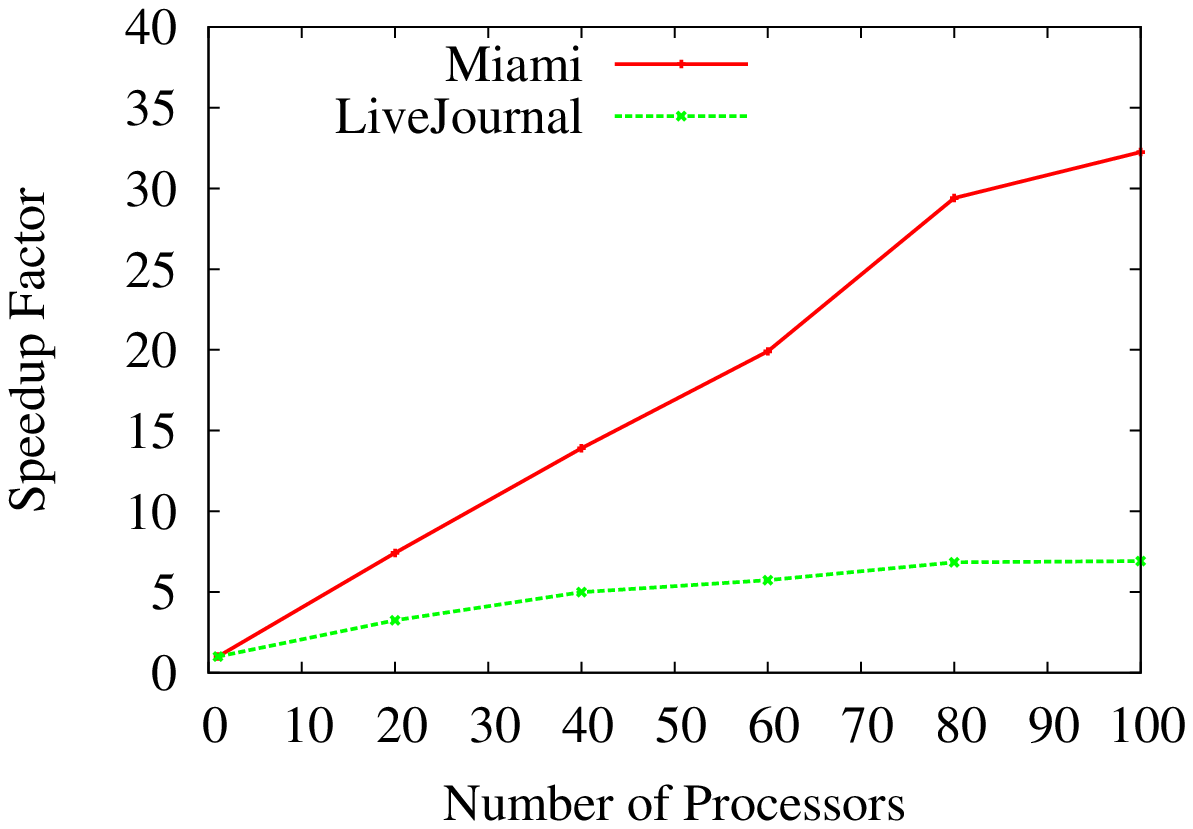}}
\fsc
\caption{Speedup with equal number of core nodes in all processors on two networks--Miami and LiveJournal.}
\label{fig:speedup_initial}
\afsc
\end{center}
\end{minipage}
\hfill
\begin{minipage}[t]{.32\textwidth}
\begin{center}
\centerline{\includegraphics[width=1.0\textwidth]{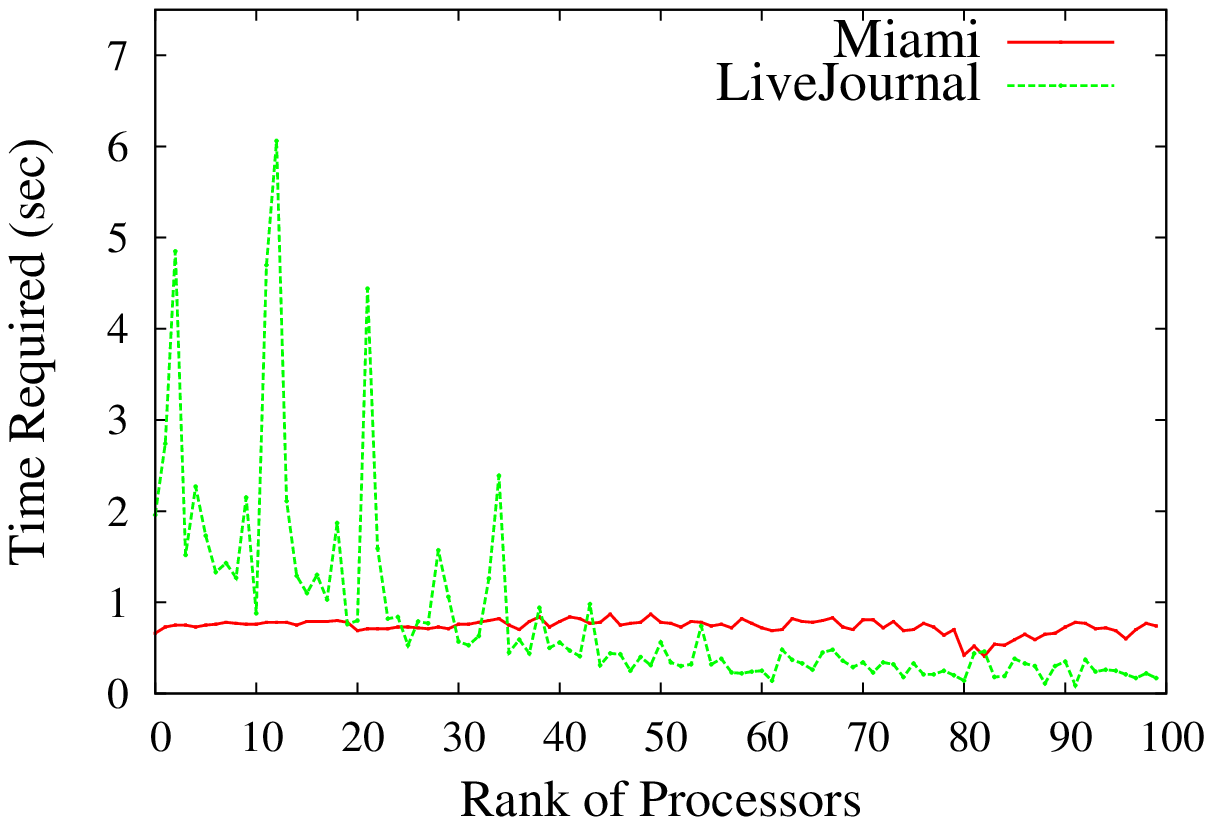}}
\fsc
\caption{Runtime of individual processors for equal number of core nodes on Miami and LiveJournal networks.}
\label{fig:ljmiami_rankwise}
\afsc
\end{center}
\end{minipage}
\hfill
\begin{minipage}[t]{.32\textwidth}
\begin{center}
\centerline{\includegraphics[width=1.0\textwidth]{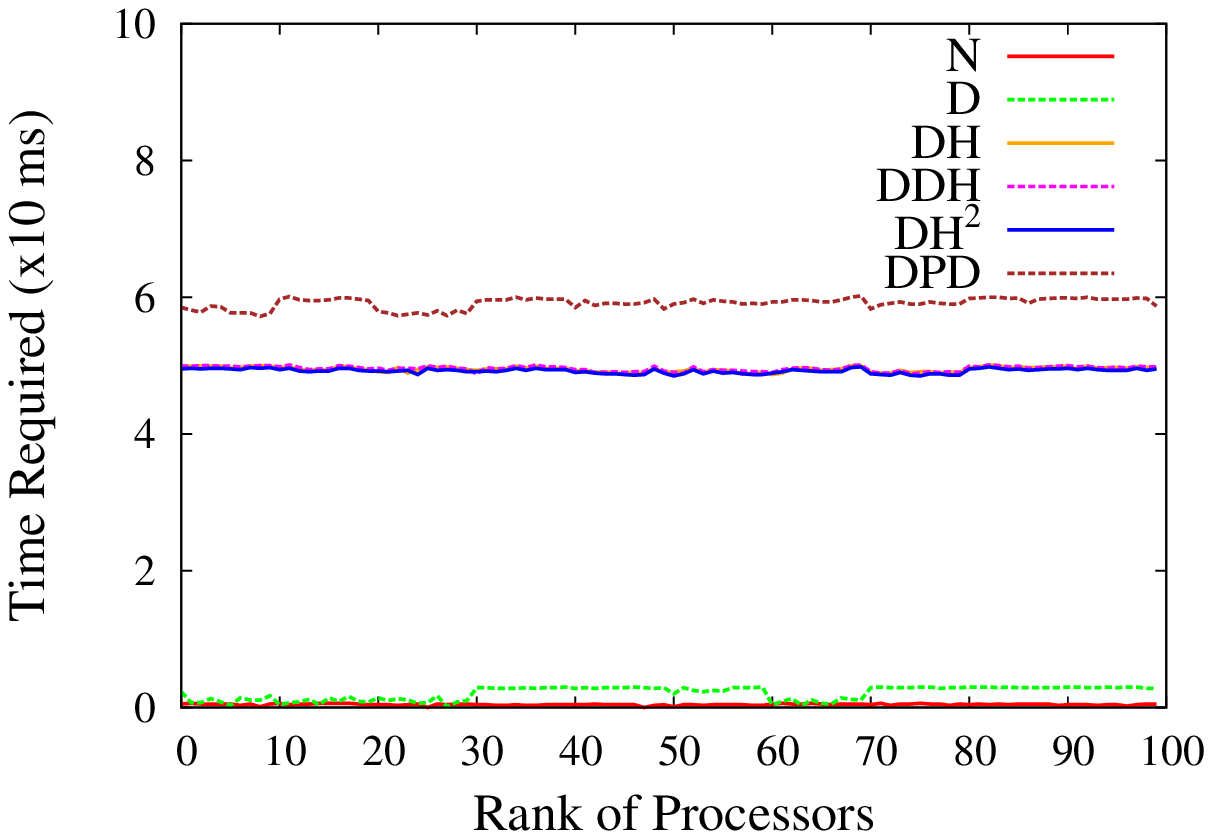}}
\fsc
\caption{Load balancing cost for LiveJournal network with different schemes.}
\label{fig:load_lj}
\afsc
\end{center}
\end{minipage}
\hfill
\vspace{-5pt}
\end{figure*}

\begin{figure*}[ht!]
	\centering
	\subfigure[Miami network]{
		\includegraphics[width=0.31\textwidth]{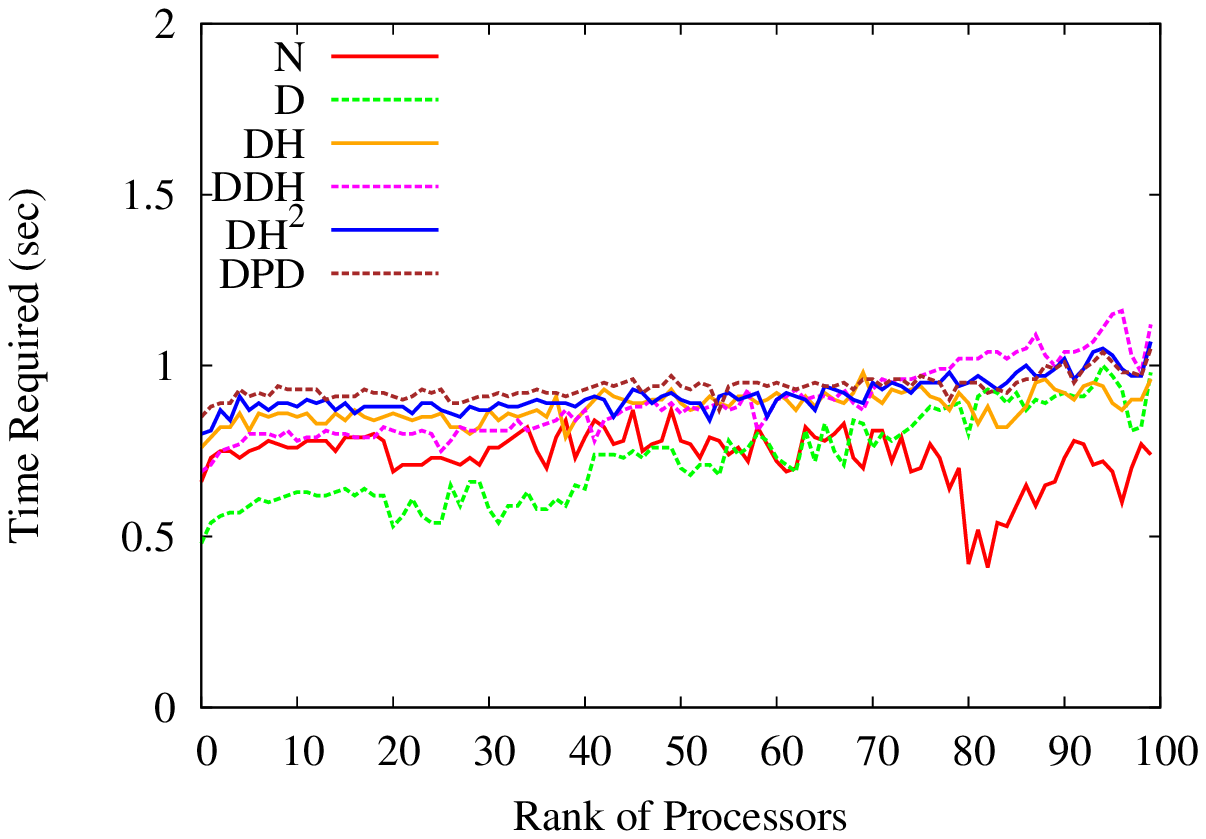}
		\label{fig:miami_total}
	}
	\subfigure[LiveJournal network]{
		\includegraphics[width=0.31\textwidth]{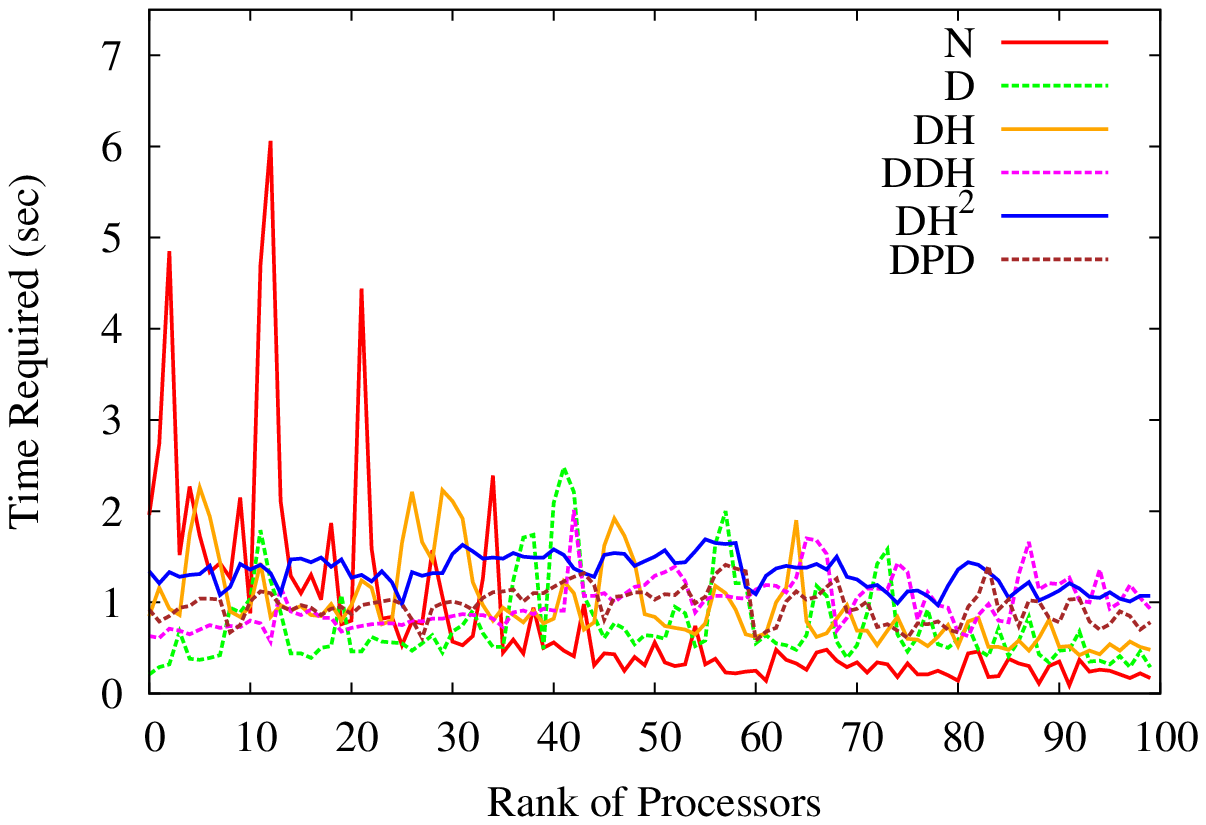}
		\label{fig:lj_total}
	}
	\subfigure[Twitter network]{
		\includegraphics[width=0.31\textwidth]{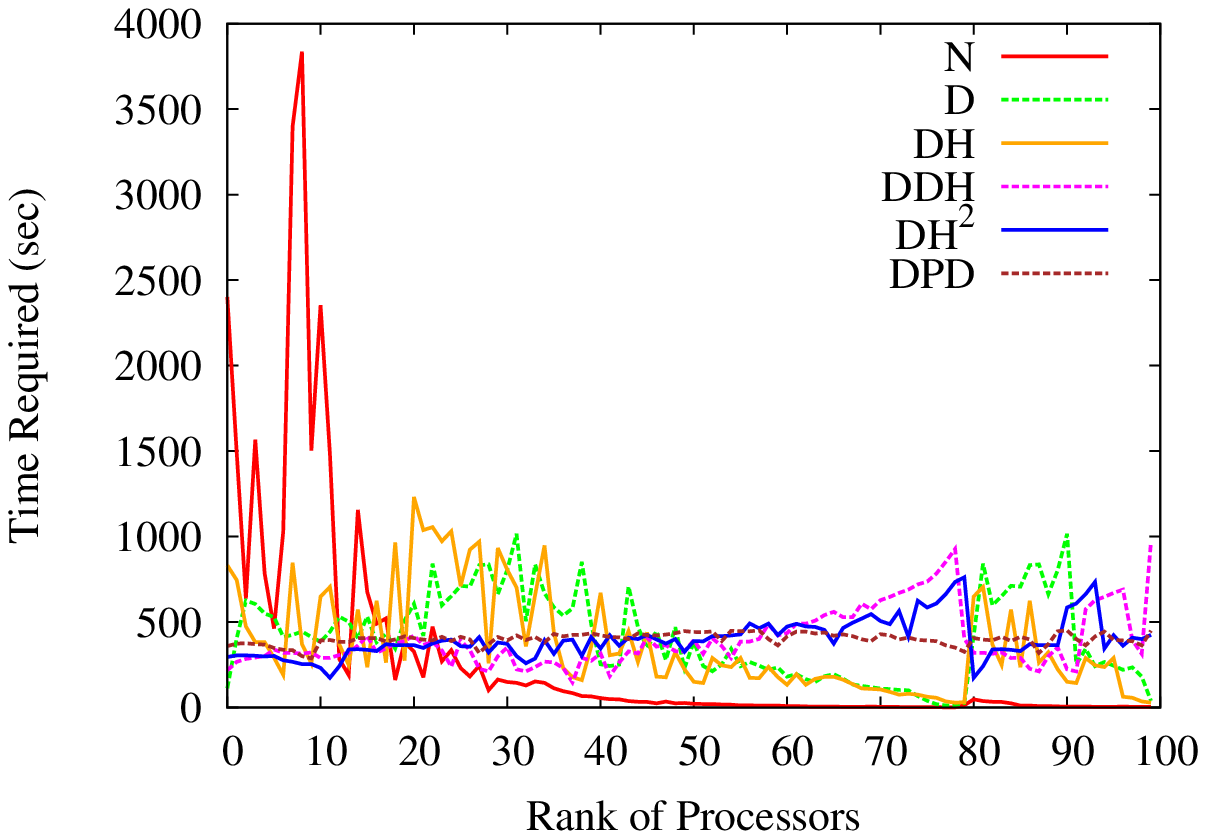}
		\label{fig:miami_total}
	}
	\fsc
	\caption{Load distribution among processors for LiveJournal, Miami, and Twitter networks by different schemes.}
	\label{fig:3net_total}
	\afsc
\end{figure*}

\subsubsection*{Proposed Load Balancing Schemes}
\label{sec:lbsec}

The balanced loads are determined before counting triangles. Thus, our parallel algorithm works in two phases: 
\begin{itemize}
\isc
\item[1.] \textit{Computing balanced load:} This phase computes partitions ${V_i^c}$ so that the computational loads are well-balanced.
\item[2.] \textit{Counting triangles:} This phase counts the triangles following the algorithms in Figure \ref{algo:patric} and \ref{algo:tcount}.
\end{itemize}

Computational cost for phase 1 is referred to as \textit {load-balancing cost}, for phase 2 as \textit{counting cost}, and the total cost for these two phases as \textit{total computational cost}. In order to be able to distribute load evenly among the processors, we need an estimation of computation load for computing triangles. For this purpose, we define a \textit{cost function} $f: V \rightarrow \mathbb{R}$, such that $f(v)$ is the computational cost for counting triangle incident on node $v$ (Lines $4$-$7$ in Figure \ref{algo:tcount}). Then, the total cost incurred to $P_i$ is given by $\sum\limits_{v \in {V_i^c}}{f(v)}$. To achieve a good load balancing, $\sum\limits_{v \in {V_i^c}}{f(v)}$ should be almost equal for all $i$. Thus, the computation of balanced load consists of the following two steps:
\begin{itemize}
\isc
\item[1.] \textbf{Computing $f$:} Compute $f(v)$ for each $v \in V$
\item[2.] \textbf{Computing partitions:} Determine $p$ disjoint partitions $V_i^c$ such that
\begin{eqnarray} \label{eqn:compload}
\sum_{v \in V_i^c}{f(v)} \approx \frac{1}{p}\sum_{v \in V}{f(v)}   \label{eqn:compload2}
\end{eqnarray}
\end{itemize}

The above computation must also be done in parallel. Otherwise, this computation takes at least $\Omega(n)$ time, which can wipe out the benefit gained from balancing load or even have a negative effect on the performance. Parallelizing the above computation, especially Step 2 (computing partitions), is a non-trivial problem. Next, we describe parallel algorithm to perform the above computation.

\vspace{0.1in}
\noindent
\textbf{Computing $f$:} \\
It might not be possible to exactly compute the value of $f(v)$ before the actual execution of counting triangles takes place. Fortunately, Theorem \ref{thm:time} provides a mathematical formulation of counting cost in terms of the number of vertices, edges, original degree $d$, and effective degree $\hat{d}$. Guided by Theorem 2, we have come up with several approximate cost function $f(v)$ which are listed in Table \ref{table:lbscheme}. Each function corresponds to one load balancing scheme. The rightmost column of the table shows identifying notations of the individual schemes.

\newcommand{\N}{\mathbb{N}}
\renewcommand{\d}{\mathbb{D}}
\renewcommand{\dh}{\mathbb{D}\mathbb{H}}
\newcommand{\dhsq}{\mathbb{D}\mathbb{H}^2}
\newcommand{\ddh}{\mathbb{D}\mathbb{D}\mathbb{H}}
\newcommand{\dpd}{\mathbb{D}\mathbb{P}\mathbb{D}}

\begin{table}[!ht]
\tbl{ Cost functions $f(.)$ for load balancing schemes.
\label{table:lbscheme}}{
    \begin{tabular}{ | l | l |}
    \hline
    {\bf Node Function} & {\bf Identifying Notation}  \\ \hline
	$f(v) = 1$   &  $\N$    \\ \hline
	$f(v) = d_v$  &  $\d$    \\ \hline
	$f(v) = \hat{d_v}$         &  $\dh$    \\ \hline
	$f(v) = d_v\hat{d_v}$   &  $\ddh$   \\ \hline
	$f(v) = \hat{d_v}^2$       &  $\dhsq$   \\ \hline
	$f(v) = \sum_{ u\in \mathcal{N}_v}{(\hat{d_v} + \hat{d_u})}$ & $\dpd$ \\ \hline
	\end{tabular}}
\end{table}

The input graph is given as a sequence of adjacency lists: adjacency list of the first node followed by that of the second node, and so on. The input sequence is considered divided by size (number of bytes) into $p$ chunks. However, it is made sure that adjacency list of a particular node reside in only one processor. Initially, processor $P_i$ stores the $i$th chunk in its memory. Let $C_i$ be the set of all nodes in the $i$-th chunk. Next, $P_i$ computes $f(v)$ for all nodes $v\in C_i$ as follows.

\vspace{-0.1in}
\begin{itemize}
\item \textbf{Scheme} $\N$: Function $f(v) = 1$ requires no computation. This scheme, essentially, assigns an equal number of core nodes to each processor.

\item \textbf{Scheme} $\d$: Function $f(v) = d_v$ requires no computation. This scheme, essentially, assigns an equal number of edges to each processor.

\item \textbf{Scheme} $\dh$: Computing function $f(v) = \hat{d_v}$ requires degrees of all $u \in \mathcal{N}_v$. Let $u\in C_j$. Then, $P_i$ sends a request message to $P_j$, and $P_j$ replies with a message containing $d_u$. 

\item \textbf{Scheme} $\ddh$: For $f(v) = d_v\hat{d_v}$, $\hat{d_v}$ is computed as above. 
\item \textbf{Scheme} $\dhsq$: For $f(v) = \hat{d_v}^2$, $\hat{d_v}$ is computed as above. 

\item \textbf{Scheme} $\dpd$: Function $f(v) = \sum_{ u\in \mathcal{N}_v}{(\hat{d_v} + \hat{d_u})}$ is computed as follows. 
\begin{itemize}
	\isc
	\item[i.] Each $P_i$ computes $\hat{d_v}$, $v\in C_i$, as discussed above.
	\item[ii.] Then $P_i$ finds $\hat{d_u}$ for all $u \in \mathcal{N}_v$: Let $u\in C_j$. $P_i$ sends a request message to $P_j$, and $P_j$ replies with a message containing $\hat{d_u}$.
	\item[iii.] Now, $f(v) = \sum_{ u\in \mathcal{N}_v}{(\hat{d_v} + \hat{d_u})}$ is computed using $\hat{d_v}$ and $\hat{d_u}$ obtained in $(i)$ and $(ii)$.  
\end{itemize}
\end{itemize}

\vspace{0.1in}
\noindent
\textbf{Computing partitions:} \\
Given that each processor $P_i$ knows $f(v)$ for all $v\in C_i$, our goal is to partition $V$ into $p$ disjoint subsets $V_i^c$ such that $\sum\limits_{v\in V_i^c}{f(v)} \approx \frac{1}{p} \sum\limits_{v\in V}{f(v)}$. 

We first compute cumulative sum $F(t)= \sum\limits_{v=0 }^{t}{f(v)}$ in parallel by using a parallel prefix sum algorithm \cite{AluruPrefix}. Processor $P_i$ computes and stores $F(t)$ for nodes $t \in C_i$. This computation takes $O\left(\frac{n}{p} + \log p\right)$ time. Notice that $P_{p-1}$ computes $F(n-1)= \sum\limits_{v=0 }^{n-1}{f(v)}$, cost for counting all triangles in the graph. $P_{p-1}$ then computes $\alpha = \frac{1}{P}\sum\limits_{v\in V}{f(v)} = \frac{1}{p}F(n-1)$ and broadcast $\alpha$ to all other processors.
Now, let $V_i^c= \{x_i, x_i+1 \dots, x_{(i+1)}-1\}$ for some node $x_i\in V$. We call $x_i$ the \textit{start} or \textit{boundary} node of partition $i$. Node $x_j$ is the $j$th boundary node if and only if $F(x_j - 1) < j\alpha \le F(x_j)$ or equivalently,  $x_j = argmin_{v \in V} \left(F(v) \ge j\alpha\right)$. A chunk $C_i$ may contain $0,1,$ or multiple boundary nodes in it. Each $P_i$ finds the boundary nodes $x_j$ in its chunk: we use the algorithm presented in \cite{alam13randomnet} to compute boundary nodes of partitions, which takes $O(n/p + p)$ time in the worst case. At the end of this execution, each processor $P_i$ knows boundary nodes $x_i$ and $x_{(i+1)}$. Now $P_i$ can construct $V_i^c$ and compute its partition $G_i(V_i,E_i)$ as described in Section \ref{sec:partition}. 

Since scheme $\dpd$ requires two levels of communication for computing $f(.)$, it has the largest load balancing cost among all schemes. Computing $f(.)$ for $\dpd$ requires $O(\frac{m}{p} + p \log p)$ time. Computing partitions has a runtime complexity of $O(\frac{m}{p} + p)$. Therefore, the load balancing cost of $\dpd$ is given by $O(\frac{m}{p} + p \log p)$. Figure \ref{fig:load_lj} shows an experimental result of the load balancing cost for different schemes on the LiveJournal network. Scheme $\N$ has the lowest cost and $\dpd$ the highest. Schemes $\dh$, $\dhsq$, and $\ddh$ have a quite similar load balancing cost. However, since scheme $\dpd$ gives the best estimation of the counting cost, it provides better load balancing. Figure \ref{fig:3net_total} demonstrates \textit{total computation cost} (load) incurred in individual processors with different schemes on Miami, LiveJournal, and Twitter networks. Miami is a network with an almost even degree distribution. Thus, all load balancing schemes, even simpler schemes like $\N$ and $\d$, distribute loads almost equally among processors. However, LiveJournal and Twitter have a very skewed degree distribution. As a result, partitioning the network based on number of nodes ($\N$) or degree ($\d$) do not provide good load balancing. The other schemes capture the computational load more precisely and produce a very even load distribution among processors. In fact, for such networks, scheme $\dpd$ provides the best load balancing.

%% file: wheel_new.pstex_t
\begin{picture}(0,0)%
\includegraphics{wheel_new.pstex}%
\end{picture}%
\setlength{\unitlength}{3522sp}%
\begingroup\makeatletter\ifx\SetFigFont\undefined%
\gdef\SetFigFont#1#2#3#4#5{%
  \reset@font\fontsize{#1}{#2pt}%
  \fontfamily{#3}\fontseries{#4}\fontshape{#5}%
  \selectfont}%
\fi\endgroup%
\begin{picture}(3983,3263)(1921,-3541)
\put(3826,-3481){\makebox(0,0)[lb]{\smash{{\SetFigFont{20}{24.0}{\rmdefault}{\mddefault}{\updefault}{\color[rgb]{0,0,0}$\dots$}%
}}}}
\put(2296,-3436){\makebox(0,0)[lb]{\smash{{\SetFigFont{20}{24.0}{\rmdefault}{\mddefault}{\updefault}{\color[rgb]{0,0,0}$v_{n-1}$}%
}}}}
\put(1936,-2356){\makebox(0,0)[lb]{\smash{{\SetFigFont{20}{24.0}{\rmdefault}{\mddefault}{\updefault}{\color[rgb]{0,0,0}$v_1$}%
}}}}
\put(2521,-691){\makebox(0,0)[lb]{\smash{{\SetFigFont{20}{24.0}{\rmdefault}{\mddefault}{\updefault}{\color[rgb]{0,0,0}$v_2$}%
}}}}
\put(3916,-2491){\makebox(0,0)[lb]{\smash{{\SetFigFont{20}{24.0}{\rmdefault}{\mddefault}{\updefault}{\color[rgb]{0,0,0}$v_0$}%
}}}}
\put(5221,-3481){\makebox(0,0)[lb]{\smash{{\SetFigFont{20}{24.0}{\rmdefault}{\mddefault}{\updefault}{\color[rgb]{0,0,0}$v_5$}%
}}}}
\put(5761,-2491){\makebox(0,0)[lb]{\smash{{\SetFigFont{20}{24.0}{\rmdefault}{\mddefault}{\updefault}{\color[rgb]{0,0,0}$v_4$}%
}}}}
\put(5221,-691){\makebox(0,0)[lb]{\smash{{\SetFigFont{20}{24.0}{\rmdefault}{\mddefault}{\updefault}{\color[rgb]{0,0,0}$v_3$}%
}}}}
\end{picture}%

%% file: performance-aop.tex

\begin{figure*}[ht!]
	\centering
		\subfigure[Miami network]{
		\includegraphics[width=0.30\columnwidth]{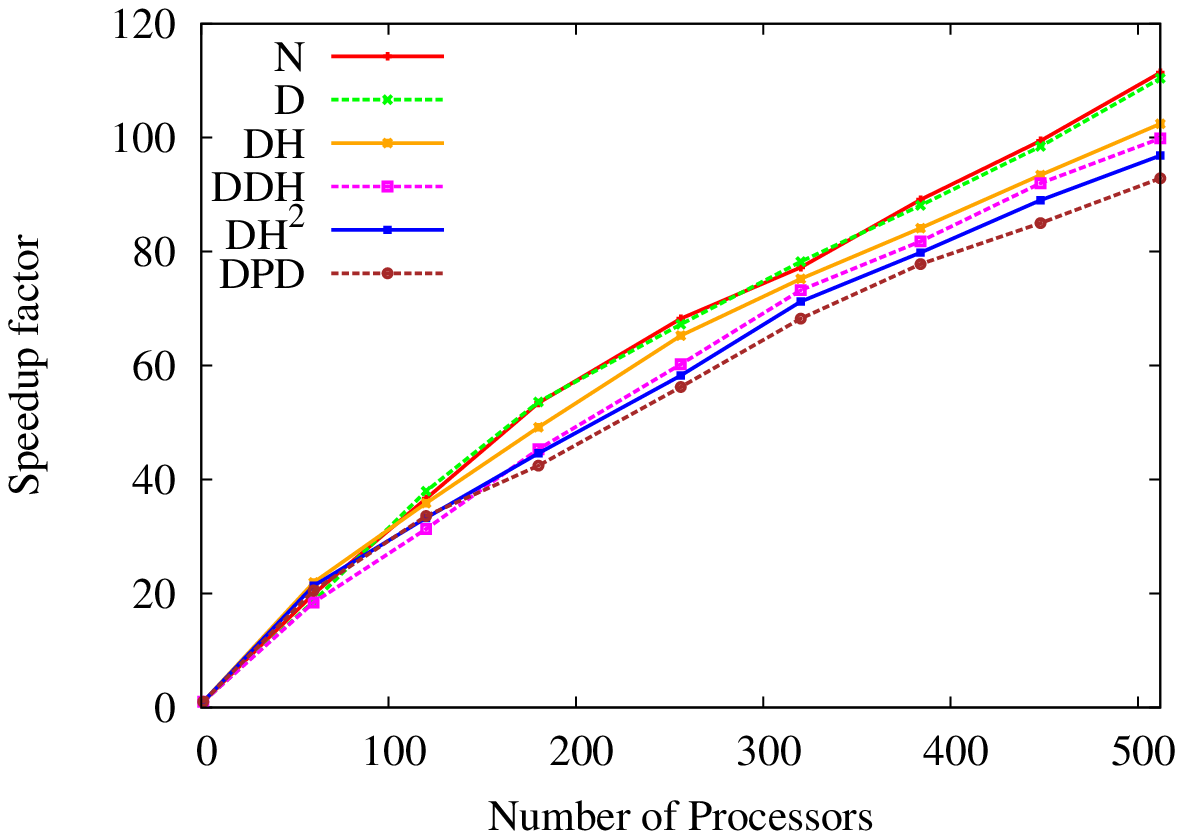}
		\label{fig:miami_speedallv}
		}
		\subfigure[LiveJournal network]{
			\includegraphics[width=0.30\columnwidth]{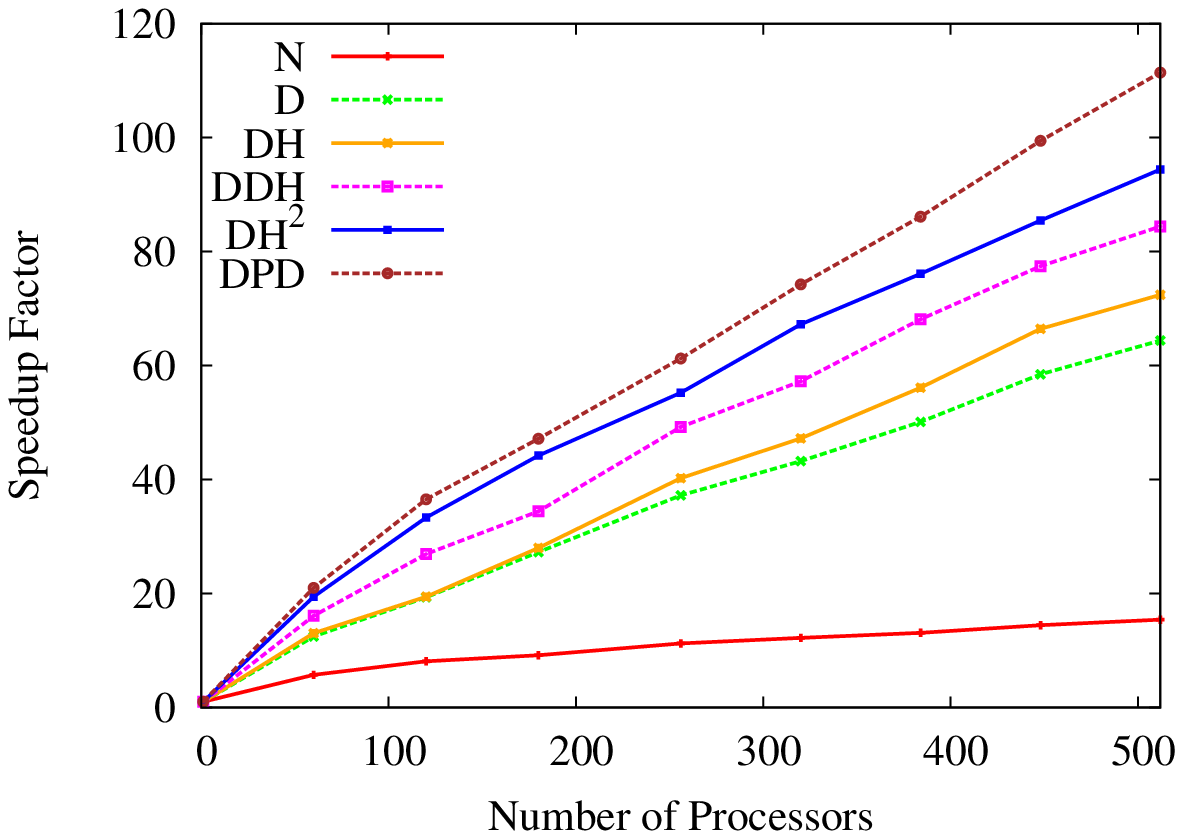}
			\label{fig:lj_speedallv}
		}
		\subfigure[Twitter network]{
			\includegraphics[width=0.30\columnwidth]{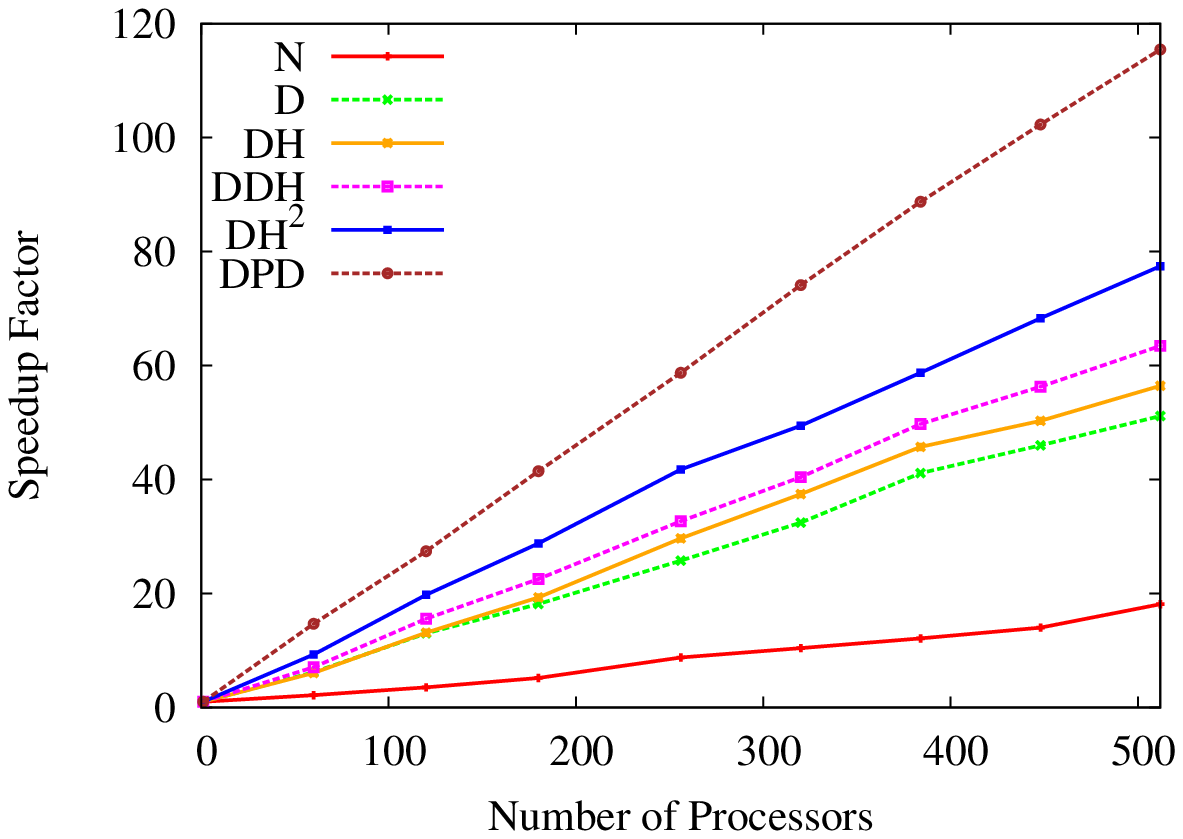}
			\label{fig:twitter_speedallv}
		}
	\fsc
	\caption{Speedup gained from different load balancing schemes for LiveJournal, Miami and Twitter networks. }
	\label{fig:3net_speedallv}
	\afsc
\end{figure*}

\subsection{Performance Analysis}
\label{sec:patric_perf}

In this section, we present the experimental results evaluating the performance of our algorithm and the load balancing schemes. 

\subsubsection{Strong Scaling} \label{subsec:strongscaling}

Strong scaling of a parallel algorithm shows how much speedup a parallel algorithm gains as the number of processors increases. Figure \ref{fig:3net_speedallv} shows strong scaling of our algorithm on LiveJournal, Miami and Twitter networks with different load balancing schemes. The speedup factors of these schemes are almost equal on Miami network. Schemes $\N$ and $\d$ have a little better speedup than the others. On the contrary, for LiveJournal and Twitter networks, speedup factors for different load balancing schemes vary quite significantly. Scheme $\dpd$ achieves better speedup than other schemes. As discussed before, for Miami network, all load balancing schemes distribute loads equally among processors. This produces an almost same speedup on Miami network with all schemes. A lower load balancing cost of schemes $\N$ and $\d$ (Figure \ref{fig:load_lj}) yields a little higher speedup. However, for LiveJournal and Twitter networks, scheme $\dpd$ gives the best load distribution (Figure \ref{fig:3net_total}) and thus provides the best speedups. Although $\dpd$ has a higher load balancing cost than others, the benefit gained from $\dpd$ as an even load distribution outweighs this cost. Thus we recommend for using $\dpd$ on real-world big graphs. Our subsequent results will be based on scheme $\dpd$.

\subsubsection{Weak Scaling}
Weak scaling of a parallel algorithm shows the ability of the algorithm to maintain constant computation time when the problem size grows proportionally with the increasing number of processors. We use PA($n, m$) networks for this experiment, and for $x$ processors, we use network PA($x/10 \times 1M, 50$). The weak scaling of our algorithm is shown in Figure \ref{fig:weak_scaling_tc}. Triangle counting cost remains almost constant (blue line). Since the load-balancing step has a communication overhead of $O(p\log p)$, load-balancing cost increases gradually with the increase of processors. It causes the total computation time to grow slowly with the addition of processors (red line). Since the growth is very slow and the runtime remains almost constant, the weak scaling of our algorithm is very good.

\subsubsection{Comparison with Previous Algorithms}

\begin{table}[tb!]
\tbl{Runtime Performance of our fast parallel algorithm using 200 processors and the algorithm in \cite{SURI}.
\label{table:patric}}{
\begin{tabular}{|l|c|c|c|} \hline
\multirow{2}{*}{ \textbf{Networks}} &  \multicolumn{2}{|c|}{\textbf{Runtime (sec.)}} & \multirow{2}{*}{\textbf{Triangles}} \\
\cline{2-3}
 & Our algorithm & \cite{SURI} &  \\ \hline
Twitter   &      $9.4$m &  423m &  $34.8$B  \\ \hline 
	web-BerkStan  &    $0.10$s  & 1.70m  & $65$M    \\ \hline 
	LiveJournal   &    $0.8$s    & 5.33m  & $286$M   \\ \hline 
	Miami         &    $0.6$s     & -- & $332$M    \\ \hline 
	PA(1B, 20)         &    $15.5$m     & -- &  $0.403$M  \\ \hline 
    \end{tabular}}
\end{table}

The runtime of our algorithm on several real and artificial networks are shown in Table
\ref{table:patric}. We also compare our algorithm with another distributed-memory parallel algorithm for counting triangles given in \cite{SURI}. We select three of the five networks used in
\cite{SURI}. Twitter and LiveJournal are the two largest among the
networks used in \cite{SURI}. We also use web-BerkStan which has a
very skewed degree distribution. No artificial network is used in
\cite{SURI}. For all of these three networks, our algorithm is more than 45
times faster than the algorithm in \cite{SURI}. The improvement over~\cite{SURI} 
is due to the fact that their algorithm generates a huge
volume of intermediate data, which are all possible 2-paths centered
at each node. The amount of such intermediate data can be
significantly larger than the original network. For example, for the
Twitter network, 300B 2-paths are generated while there are only 2.4B
edges in the network. The algorithm in \cite{SURI} shuffles and
regroups these 2-paths, which take significantly larger time and also
memory. 

\begin{figure*}[!tbh]
	\hfill
	\begin{minipage}[t]{.47\textwidth}
		\begin{center}
			\centering
			\includegraphics[scale=0.5]{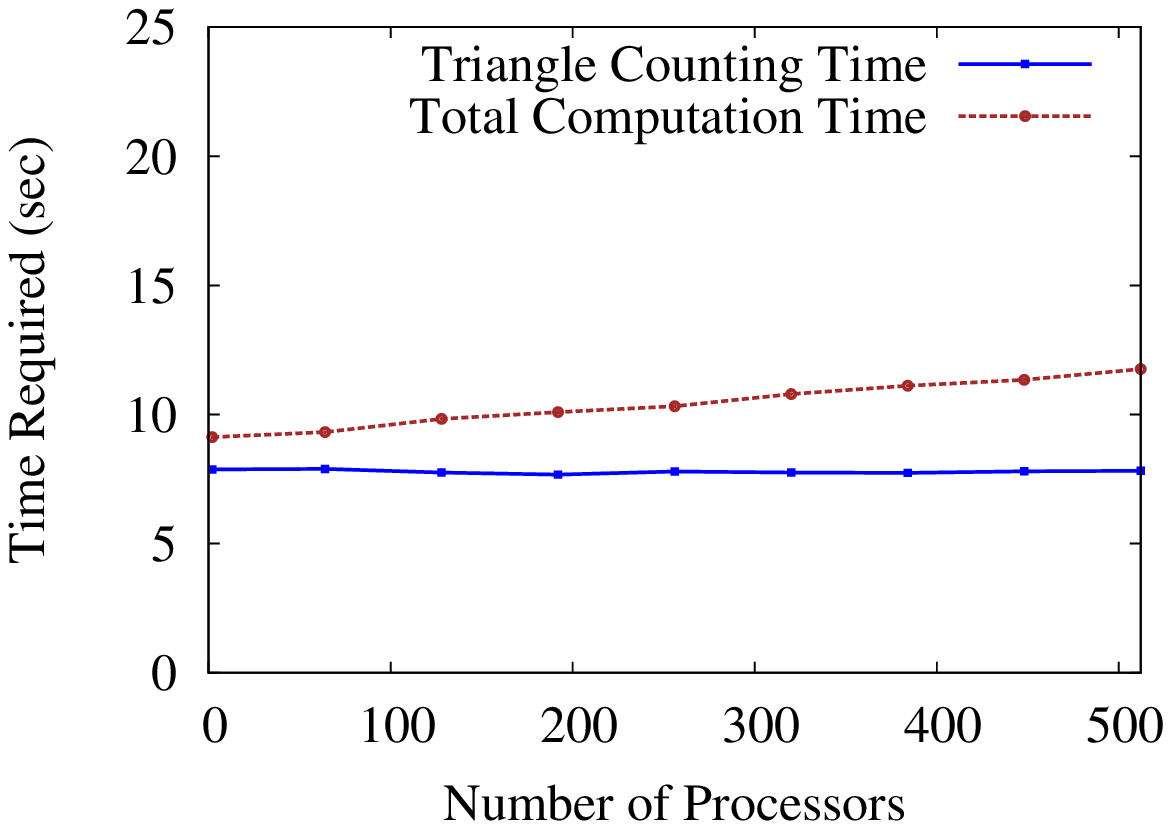}
			\caption{Weak scaling on PA($p/10 \times 1M, 50$) networks.}
			\label{fig:weak_scaling_tc}
		\end{center}
	\end{minipage}
	\hfill
	\begin{minipage}[t]{.47\textwidth}
		\begin{center}
			\centering
			\includegraphics[scale=0.5]{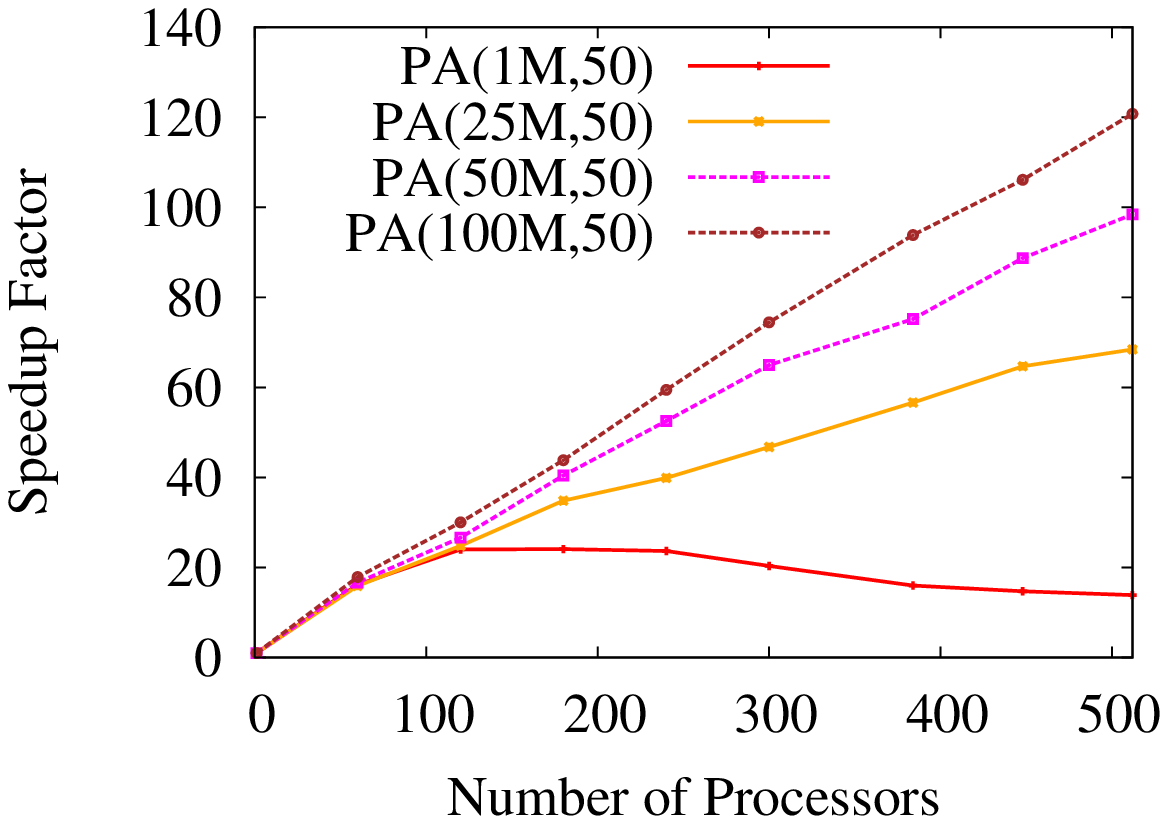}
			\caption{Improved scalability with increased network size.}
			\label{fig:size_scale}
		\end{center}
	\end{minipage}
	\hfill
	\vspace{-5pt}
\end{figure*}

\subsubsection{Scaling with Network Size} \label{subsec:net_size_scaling}

The load-balancing cost of our algorithm, as shown in Section \ref{sec:load}, is $O(m/p + p\log p)$ where $p$ is the number of processors used in the computation. For the algorithm given in Figure \ref{algo:tcount}, the counting cost is $O(\sum_{v\in V_i^c}\sum_{u\in N_v}{(\hat{d}_u + \hat{d}_v)})$. Thus, the total computational cost of our algorithm is,
\begin{eqnarray*} \label{eqn:totalcost}
F(p) &=& O(\frac{m}{p} + p\log p + \max_i \sum_{v\in V_i^c}\sum_{u\in N_v}{(\hat{d}_u + \hat{d}_v)}) \\
 & \approx & c_1 \frac{m}{p} + c_2 p\log p + c_3\max_i \sum_{v\in V_i^c}\sum_{u\in N_v}{(\hat{d}_u + \hat{d}_v)},
\end{eqnarray*}
where $c_1$, $c_2$, and $c_3$ are  constants. Now, quantity denoting computation cost, ($c_1 m/p + c_3 \sum_{v\in V_i^c}\sum_{u\in N_v}{(\hat{d}_u + \hat{d}_v)}$), decreases with the increase of $p$, but communication cost $p\log p$ increases with $p$. Thus, initially when $p$ increases, the overall runtime decreases (hence the speedup increases). But, for some large value of $p$, the term $p \log p$ becomes dominating, and the overall runtime increases with the addition of further processors. Notice that communication cost $p\log p$ is independent of network size. Therefore, when networks grow larger, computation cost increases, and hence they scale to a higher number of processors, as shown in Figure \ref{fig:size_scale}. This is, in fact, a highly desirable behavior of our parallel algorithm which is designed for real world massive networks. We need large number of processors when the network size is large and computation time is high.

Consequently, there is an optimal value of $p$, $p_{opt}$, for which the total time $F(p)$ drops to its minimum and the speedup reaches its maximum. To have an  estimation of $p_{opt}$, we replace $d$ and $\hat{d}$ with average degree $\bar{d}$ and $\bar{d}/2$, respectively, and have $F(p)  \approx  c_1n\bar{d}/p + c_2p\log p + c_3n\bar{d}^2/p$. At the minimum point, $\frac{\mathrm d}{\mathrm d p} \big( F(p) \big) = 0$, which gives the following relationship of $p_{opt}$, $n$ and $\bar{d}$: $p^2(1 + \log p) = \frac{n}{c_2} (c_3\bar{d}^2 + c_1\bar{d})$. Thus, $p_{opt}$ has roughly a linear relationship with $\sqrt{n}$ and $\bar{d}$. 

Assume that a network with the number of nodes $n'$ and average degree $\bar{d'}$ experimentally shows an optimal $p$ of $p'_{opt}$. Then, another network with $n$ nodes and an average degree $\bar{d}$ has an approximate optimum number of processors, 
\begin{equation} \label{eqn:estimation}
p_{opt} \approx p'_{opt}\frac{\bar{d}}{\bar{d'}}\sqrt{\frac{n}{n'}}.
\end{equation}

Thus, if we compute $p'_{opt}$ experimentally by trial and error for an available network (let's call it the \textit{base network}), we can estimate $p_{opt}$ for all other networks. The base network might be a small network for which this trial-error should be fairly fast. From the result presented in Figure \ref{fig:size_scale}, the network $PA(1M,50)$ can serve as a base network, and $p_{opt}$ for the network $PA(25M,50)$ can be estimated as $p_{opt}\approx 600$ which is approximately $5$ times of that of $PA(1M,50)$ ($p'_{opt}\approx 120$). The relationship is also justified when we vary average degree of the networks.

%% file: background-anop.tex
\section{A Space-efficient Parallel Algorithm with Non-overlapping Partitioning}
\label{sec:space}

The algorithm presented in Section \ref{sec:patric} divides the input graph into a set of $p$ overlapping partitions where some edges $(u,v)$ might be repeated (overlapped) in multiple partitions. Such overlapping allows the algorithm to count triangles without any communication among processors leading to faster computation. Further, since each processor works on a part of the entire graph, the algorithm can work on large graphs.
However, for instances where the graph has a high average degree or a few nodes with high degrees, overlapping partitions can be large. Now, if overlapping of edges among partitions are avoided, we can further improve the space efficiency of the algorithm. In this section, we present a parallel algorithm which divides the input graph into non-overlapping partitions. Each edge resides in a single partition, and the sizes of all partitions sum up to the size of the graph. Non-overlapping partitioning leads to a more space efficient algorithm and thus allows to work on larger graphs. In fact, non-overlapping partitioning offers as much as $\bar{d}$ (average degree of the graph) times space saving over the overlapping partitions. Table \ref{table:mem-comp} shows the space requirement of non-overlapping partitions which is up to $25$ times smaller than that overlapping partitions for the networks we experimented on.  

\begin{table}[!htb]
	\tbl{Memory usage of our algorithms (size of the largest partition) with both overlapping and non-overlapping partitioning. Number of partitions used is $100$.
	\label{table:mem-comp}}{
	\begin{tabular}{| l | l | l | l | l | l |} \hline
		\multirow{2}{*}{ \textbf{Networks}} &  \multicolumn{2}{|c|}{\textbf{Memory (MB)}} & \multirow{2}{*}{\textbf{Ratio}} & \multirow{2}{*}{\textbf{$\bar{d}$}} & \multirow{2}{*}{\textbf{$d_{max}$}}\\
		\cline{2-3}
		& Non-overlap. & Overlap. &  & &\\ \hline
		web-Google   &    $1.49$    & $11.3$  & $7.85$ &$11.6$ & $6332$ \\ \hline 
		LiveJournal   &    $9.41$    & $110.75$  & $11.75$ &$18$ &  $20333$ \\ \hline 
		Miami         &    $10.63$     & $109.58$ & $10.32$ & $47.6$ &  $425$ \\ \hline 
		Twitter   &      $265.82$ &  $4254.18$ & $16.004$ &$57.1$ & $1001159$ \\ \hline 
		PA(10M, 100)         &    $121.11$     & $2120.94$ & $17.5$ &$100$ & $25068$\\ \hline 
		PA(1M, 1000)         &    $138.20$     & $3427.36$ & $24.8$ &$1000$ & $19255$\\ \hline 
	\end{tabular}}
\end{table}

Notice the space requirement of the other distributed-memory parallel algorithms for counting the exact number of triangles in literature: the first MapReduce based algorithm proposed in \cite{SURI} generates a huge amount of intermediate data which is significantly larger than the original network (e.g., $125$ times larger for Twitter network). The second MapReduce based algorithm proposed in \cite{SURI}, the partition-based algorithm, has a space requirement of $O(mp)$ for the Map phase (with $p$ partitions), which is $p$ times larger than the network size. The algorithm in \cite{park-cikm} also requires $O(mp)$ memory space. Our space-efficient algorithm requires only a total of $O(m)$ space for storing all $p$ partitions.

%% file: parallelalgo-anop.tex
\subsection{Overview of Our Space-Efficient Parallel Algorithm}
\label{sec:space-ov}
This algorithm partitions the input graph $G(V,E)$ into a set of $p$ partitions constructed as follows: set of nodes $V$ is partitioned into $p$ disjoint subsets $V_i^c$, such that, for $0 \leq j,k \leq p-1$ and $j \neq k$, $V_j^c \cap V_k^c = \emptyset$ and $\bigcup_{k} V_k^c = V$. Edge set $E_i^c$, constructed as $E_i^c = \{(u,v):u\in V_i^c, v\in N_u\}$, constitutes the $i$-th partition. Note that this partition is non-overlapping-- each edge $(u,v) \in E$ resides in one and only one partition. For $0 \leq j,k \leq p-1$ and $j \neq k$, $E_j^c \cap E_k^c = \emptyset$ and $\bigcup_{k} E_k^c = E$. The sum of space required to store all partitions equals to the space required to store the whole graph. 

Now, to count triangles incident on $v \in V_i^c$, processor $P_i$ needs $N_u$ for all $u \in N_v$ (Lines 7-10, Fig. \ref{algo:nodeiteratorn}). If $u\in V_i^c$, information of both $N_v$ and $N_u$ is available in the $i$-th partition, and $P_i$ counts triangles incident on $(v,u)$ by computing $N_u \cap N_v$. However, if $u \in V_j^c$, $j \neq i$, $N_u$ resides in partition $j$. Processor $P_i$ and $P_j$ exchange message(s) to count triangles incident on such $(v,u)$. This exchanging of messages introduces a communication overhead which is a crucial factor on the performance of the algorithm. We devise an efficient approach to reduce the communication overhead drastically and improve the performance significantly. Once all processors complete the computation associated with respective partitions, the counts from all processors are aggregated.

\subsection{An Efficient Communication Approach}
\label{sec:space-tc}
Processor $P_i$ and $P_j$ require to exchange messages for counting triangles incident on $(v,u)$ where $v\in V_i^c$ and $u\in N_v \cap V_j^c $. A simple way to count such triangles is as follows: $P_i$ requests $P_j$ for $N_u$. $P_j$ sends $N_u$ to $P_i$, and $P_i$ counts triangles incident on the edge $(v,u)$ by computing $N_v \cap N_u$. For further reference, we call this approach as \textit{direct approach}. This approach requires exchanging as much as $O(m\bar{d})$ messages ($\bar{d}$ is the average degree of the network) which is substantially larger than the size of the graph. 

The above approach has a high communication overhead due to exchanging a large number of redundant messages leading to a large runtime. Assume $u\in N_{v_1}\cap N_{v_2}\cap\dots\cap N_{v_k}$, for $v_1, v_2, \dots, v_k \in V_i^c$. Then $P_i$ sends $k$ separate requests for $N_u$ to $P_j$ while computing triangles incident on $v_1$, $v_1$, $\dots$, $v_k$. In response to those requests, $P_j$ sends $N_u$ to $P_i$ $k$ times. 

One seemingly obvious way to eliminate redundant messages is that instead of requesting $N_u$ multiple times, $P_i$ stores it in memory for subsequent use. However, space requirement for storing all $N_u$ along with the partition $i$ itself is the same as that of storing an overlapping partition. This diminishes our original goal of a space-efficient algorithm.

Another way of eliminating message redundancy is as follows. When $N_u$ is fetched, $P_i$ completes all computation that requires $N_u$: $P_i$ finds all $k$ nodes $v\in V_i^c$ such that $u\in N_v$. It then performs all $k$ computations $N_v \cap N_u$ involving $N_u$ and discards $N_u$. Now, since $u\in N_v \implies v\notin N_u$, $P_i$ cannot extract all such nodes $v$ from the message $N_u$. Instead, $P_i$ requires to scan through its whole partition to find such nodes $v$ where $u \in N_v$. This \textit{scanning} is very expensive-- requiring $O(\sum_{v\in V_i^c} d_v)$ time for each message-- which might even be slower than the direct approach with redundant messages.

All the above techniques to improve the efficiency of  \textit{Direct} approach introduce additional space or runtime overhead. Below we propose an efficient approach to reduce message exchanges drastically without adding further overhead.

\textbf{Reduction of messages.} To compute $N_v\cap N_u$ for $v\in V_i^c$ and $u\in N_v\cap V_j^c$, $P_i$ requires fetching $N_u$ from partition $j$. Instead, $P_j$ can perform the same computation if $P_i$ sends $N_v$ to $P_j$. Specifically, we consider the following approach: $P_i$ sends  $N_v$ to $P_j$ instead of fetching $N_u$. $P_j$ counts triangles incident on edge $(u,v)$  by performing the operation $N_v \cap N_u$. We call this approach as \textit{Surrogate} approach.

On a surface, this approach might seem to be a simple modification from \textit{Direct} approach. However, notice the following implication which is very significant to the algorithm: once $P_j$ receives $N_v$, it can extract the information of all nodes $u$, such that $u$ is in both $N_v$ and $V_j^c$, by scanning $N_v$ only. For all such nodes $u$, $P_j$ counts triangles incident on edge $(u,v)$  by performing the operation $N_v \cap N_u$. $P_j$ then discards $N_v$ since it is no longer needed. Note that extracting all $u$ such that $u\in N_v$ and $u\in V_j$ requires $O(d_v)$ time (compare this to $O(\sum_{v\in V_i^c} d_v)$ time of direct approach for the same purpose). In fact, this extraction can be done while computing triangles $N_v\cap N_u$ for first such $u$. This saves from any additional overhead.

As we noticed, if delegated, $P_j$ can count triangles on multiple edges $(u,v)$ from a single message $N_v$, where $v\in V_i^c$ and $u\in N_v\cap V_j^c$. Thus $P_i$ does not require to send $N_v$ to $P_j$ multiple times for each such $u$. However, to avoid multiple sending, $P_i$ needs to keep track of which processors it has already sent $N_v$ to. This \textit{message tracking} needs to be done carefully, otherwise any additional space or runtime overhead might compromise the efficiency of the overall approach.

It is easy to see that one can perform the above tracking by maintaining $p$ \textit{flag} variables, one for each processor. Before sending $N_v$ to a particular processor $P_j$, $P_i$ checks $j$-th flag to see if it is already sent. This implementation is conceptually simple but cost for resetting flags for each $v\in V_i^c$ sums to a significant cost of $O(|V_i^c|.p)$. Now notice that an overhead of $O(|V_i^c|.p)$ will lead to a runtime of at least $\Omega(n)$ because $\max_i|V_i^c| \ge \frac{n}{p}$. An algorithm with $\Omega(n)$ will not be scalable to a large number of processors since with the increase of $p$, the runtime $\Omega(n)$ does not decrease. 

Now, observe the following simple yet useful property of $N_v$: \textit{Since $V_j^c$ is a set of consecutive nodes, and all neighbor lists $N_v$ are sorted, all nodes $u \in N_v\cap V_j^c$ reside in $N_v$ in consecutive positions.} This property enables each $P_i$ to track messages by only recording the last processor (say, \textit{LastProc}) it has sent $N_v$ to. When $P_i$ encounters $u\in N_v$ such that $u\in V_j^c$, it checks \textit{LastProc}. If $\textit{LastProc} \neq P_j$, then $P_i$ sends $N_v$ to $P_j$ and set $\textit{LastProc} = P_j$. Otherwise, the node $u$ is ignored, meaning it would be redundant to send $N_v$. Resetting a single variable $\textit{LastProc}$ has a overhead of $O(|V_i^c|)$ as opposed to $O(|V_i^c|.p)$. 

Thus surrogate approach detects and eliminates message redundancy and allows multiple computation from a single message, without even compromising execution or space efficiency. The efficiency gained from this capability is shown experimentally in Section \ref{sec:perf}.

\subsection{Pseudocode for Counting Triangles.}

We denote a message by $\left\langle t, X \right\rangle $ where $t \in \{data, control\}$ is the type and $X$ is the actual data associated with the message. For a data message ($t=data$), $X$ refers to a neighbor list $N_x$ whereas for a control ($t=control$), $X=\emptyset$. The pesudocode for counting triangles for an incoming data message $\left\langle data, X \right\rangle$ is given in Fig. \ref{algo:space-tc-msg}.

\begin{figure}[!ht]
\begin{center}
\fbox{
\begin{minipage}[c] {0.85\linewidth}
\begin{algorithmic}[1]
	\STATE {Procedure \textsc{SurrogateCount}$(X,i):$}
	\STATE {$T \leftarrow 0$}\hspace{0.1in} //$T$ is the count of triangles
	\FOR {all $u\in X$ such that $u \in V_i^c$ }
			\STATE $S \leftarrow N_u \cap X$
			\STATE $T \leftarrow T+|S|$ 
	\ENDFOR
	\RETURN $T$
\end{algorithmic}
\end{minipage}
}
\end{center}
\vspace{-0.1in}
\caption{The procedure executed by $P_i$ after receiving message $\left\langle data, X \right\rangle $ from some $P_j$.}
\label{algo:space-tc-msg}
\vspace{-0.1in}
\end{figure} 

Once a processor $P_i$ completes the computation on all $v\in V_i^c$, it broadcasts a completion message $\left\langle control, \emptyset \right\rangle $. However, it cannot terminate execution until it receives $\left\langle control, \emptyset \right\rangle $ from all other processors since other processors might send data messages for surrogate computation. Finally, $P_0$ sums up counts from all processors using MPI aggregation function. The complete pseudocode of our algorithm using surrogate approach is presented in Fig. \ref{algo:space-tc}.

\begin{figure}[!ht]
\begin{center}
\fbox{
\begin{minipage}[c] {0.95\linewidth}
\begin{algorithmic}[1]
\STATE {$T_i \leftarrow 0$}\hspace{0.1in} //$T_i$ is $P_i$'s count of triangles
\FOR {each $v \in V_i^c$}
    \FOR {$u \in N_v$}
	    \IF { $u \in V_i^c$}
			\STATE $S \leftarrow N_v \cap N_u$
			\STATE $T_i \leftarrow T_i+|S|$ 
		\ELSE
			\STATE {\textbf{Send} $\left\langle data, N_v \right\rangle $ to $P_j$, where $u \in V_j$, if not sent already}	
		\ENDIF
	\ENDFOR
	\STATE {}
	\FOR {each incoming message $\left\langle t, X \right\rangle $}
	\IF { $t=data $}
	\STATE {$T_i \leftarrow T_i +$ \textsc{SurrogateCount}$(X,i)$} // See Figure \ref{algo:space-tc}
	\ELSE 
	\STATE {\textbf{Increment} completion counter}	
	\ENDIF
	\ENDFOR	
\ENDFOR
\STATE{}
\STATE{\textbf{Broadcast} $\left\langle control, \emptyset \right\rangle $}
\WHILE{completion counter $<$ p-1} 
	\FOR {each incoming message $\left\langle t, X \right\rangle $}
	\IF { $t=data $}
	\STATE {$T_i \leftarrow T_i +$ \textsc{SurrogateCount}$(X,i)$} // See Figure \ref{algo:space-tc}
	\ELSE 
	\STATE {\textbf{Increment} completion counter}	
	\ENDIF
	\ENDFOR
\ENDWHILE
\STATE{}
\STATE \textsc{MpiBarrier}
\STATE \textbf{Find Sum} $T \leftarrow \sum_i{T_i}$ using $\textsc{MpiReduce}$

\end{algorithmic}
\end{minipage}
}
\end{center}
\fsc
\caption{An algorithm for counting triangles using surrogate approach. 
Each processor $P_i$ executes Line 1-22. After that, they are synchronized, and the aggregation is performed (Line 24-25).}
\label{algo:space-tc}
\afsc
\end{figure}

 \subsection{Partitioning and Load Balancing}
 \label{sec:space-part}
 
 While constructing partitions $i$, set of nodes $V$ is partitioned into $p$ disjoint subsets $V_i^c$ of consecutive nodes. 
 Ideally, the set $V$ should be partitioned in such a way that the cost for counting triangles is almost equal for all processors. 
 Similar to our fast parallel algorithm presented in Section \ref{sec:patric}, we need to compute $p$ disjoint partitions of $V$ such that for each partition $V_i^c$, 
 \begin{eqnarray} \label{eqn:part}
 \sum_{v \in V_i^c}{f(v)} \approx \frac{1}{p}\sum_{v \in V}{f(v)}.
 \end{eqnarray}

 Several estimations for $f(v)$ were proposed in Section \ref{sec:patric} among which $f(v) = \sum_{ u\in N_v}{(\hat{d_v} + \hat{d_u})}$ was shown experimentally as the best. Since our algorithm employs a different communication scheme for counting triangles, none of those estimations corresponds to the cost of this algorithm. Thus, we derive a new cost function $f(v)$ to estimate the computational cost of our algorithm more precisely. 

 \textbf{Deriving An Estimation for Cost Function $f(v)$.}
  
 We want to find $f(v)$ such that $\sum_{v\in V_i^c} f(v)$ gives a good estimation of the computation cost incurred on processor $P_i$. We derive $f(v)$ as follows.

 Recall that $\mathcal{N}_v = \{u: (u,v)\in E \}$ and $N_v = \{u: (u,v)\in E, v \prec u \}$. Then, it is easy to see that 
 \begin{eqnarray} \label{eqn:nv}
 u \in \mathcal{N}_v-N_v \Leftrightarrow v \in N_u.	
 \end{eqnarray}
 Now, $P_i$ performs two types of computations due to all $v \in V_i^c$ as follows.
 
 \begin{itemize}[leftmargin=*] 
 	\item[1.] \textit{Surrogate or delegated computation:} 
 	$P_i$ compute $N_v \cap N_u$ for all $v \in N_u$ and $u \in V_j^c$, $i \ne j$, i.e., $u\in (\mathcal{N}_v-N_v) \cap (V-V_i^c)$. The cost incurred on $P_i$ for such $u$ and $v$ is given by
 	\begin{eqnarray*} \label{eqn:surr_comp}
 		\Theta\left(\sum_{v\in V_i^c}\sum_{ u\in (\mathcal{N}_v-N_v) \cap (V-V_i^c)}{(\hat{d_v} + \hat{d_u})}\right).
 	\end{eqnarray*}
 	
 	\item[2.] \textit{Local computation:} 
 	$P_i$ compute $N_v \cap N_u$ for all $u \in N_v \cap V_i^c$. Let $E_i^c$ be the set of edges $(u,v)$ where both $u$ and $v$ are in $V_i^c$, i.e., $E_i^c=\{ (u,v) \in E | u,v \in V_i^c\}$. Now, the	cost incurred on $P_i$ for local computations is given by
 	\begin{eqnarray*} \label{eqn:surr_comp}
 		\Theta\left(\sum_{v\in V_i^c}\sum_{ u\in N_v \cap V_i^c}{(\hat{d_v} + \hat{d_u})}\right) 
 		&=&\Theta\left(\sum_{(u,v)\in E_i^c}{(\hat{d_v} + \hat{d_u})}\right) \\
 		&=& \Theta\left(\sum_{v\in V_i^c}\sum_{ u \in (\mathcal{N}_v-N_v) \cap V_i^c}{(\hat{d_v} + \hat{d_u})}\right).
 	\end{eqnarray*}
 \end{itemize}
 By adding costs from $(1)$ and $(2)$ above, we get the computation cost, 
 \vspace{-0.1in}
 \begin{eqnarray*}
 	\Theta \left(\sum_{v\in V_i^c}\sum_{ u\in \mathcal{N}_v-N_v}{(\hat{d_v} + \hat{d_u})}\right). 	
 \end{eqnarray*}
 
 Now, if we assign $f(v)=\left(\sum_{ u\in \mathcal{N}_v-N_v}{(\hat{d_v} + \hat{d_u})}\right)$, the computation cost incurred on $P_i$ becomes $\sum_{v\in V_i^c} f(v)$. Thus, we use the following cost function: 
 \begin{eqnarray*}
 	f(v)=\left(\sum_{ u\in \mathcal{N}_v-N_v}{(\hat{d_v} + \hat{d_u})}\right).
 \end{eqnarray*}

 \textbf{Parallel Computation of the Cost Function $f(v)$.} 
 In parallel, each processor $P_i$ computes $f(v)$ for all $v\in C_i$. Recall that $C_i$ is the set of all nodes in the $i$-th chunk, as discussed in Section \ref{sec:load}. Function $f(v) = \left(\sum_{ u\in \mathcal{N}_v-N_v}{(\hat{d_v} + \hat{d_u})}\right)$ is computed as follows. 
 \begin{itemize}
 	\isc
 	\item[i.] First $P_i$ computes $\hat{d_v}$, $v\in C_i$: computing $\hat{d_v}$ requires $d_u$ for all $u \in \mathcal{N}_v$. Let $u\in C_j$. Then, $P_i$ sends a request message to $P_j$, and $P_j$ replies with a message containing $d_u$. 
 	\item[ii.] Then $P_i$ finds $\hat{d_u}$ for all $u \in \mathcal{N}_v-N_v$: let $u\in C_j$. $P_i$ sends a request message to $P_j$, and $P_j$ replies with a message containing $\hat{d_u}$.
 	\item[iii.] Now, $f(v) = \sum_{ u\in \mathcal{N}_v-N_v}{(\hat{d_v} + \hat{d_u})}$ is computed using $\hat{d_v}$ and $\hat{d_u}$ obtained in step $(i)$ and $(ii)$.  
 \end{itemize}
 
 \textbf{Computing Balanced Partitions.}
 Once $f(v)$ is computed for all $v\in V$, we compute $V_i^c$ using the same algorithm we used for overlapping partitioning as described in Section \ref{sec:patric}. 

\subsection{Correctness of The Algorithm}
\label{sec:correct_spacealgo}

The correctness of our space efficient parallel algorithm is formally presented in the following theorem.

\begin{theorem} \label{thm:correct_spacealgo}
Given a graph $G=(V, E)$, our space efficient parallel algorithm counts every triangle in G exactly once.
\end{theorem}

\noindent \textit{Proof.} Consider a triangle $(x_1, x_2, x_3)$ in $G$, and without the loss of generality, assume that $x_1 \prec x_2 \prec x_3$. By the constructions of $N_x$ (Line 2-4 in Fig. \ref{algo:nodeiteratorn}), we have $x_2, x_3 \in {N_{x_1}}$ and $x_3 \in N_{x_2}$. Now, there are two cases: 

\begin{itemize}[leftmargin=*] 
\item \textit{case 1.} $x_1, x_2 \in V_i^c$:
Nodes $x_1$ and $x_2$ are in the same partition $i$. Processor $P_i$ executes the loop in Line 2-6 (Fig. \ref{algo:space-tc}) with $v = x_1$ and $u = x_2$, and node $x_3$ appears in $S = N_{x_1} \cap N_{x_2}$, and the triangle $(x_1, x_2, x_3)$ is counted once. But this triangle cannot be counted for any other values of $v$ and $u$ because $x_1 \notin N_{x_2}$ and $x_1,x_2 \notin N_{x_3}$.

\item \textit{case 2.} $x_1 \in V_i^c, x_2 \in V_j^c, i \neq j$:
Nodes $x_1$ and $x_2$ are in two different partitions $i$ and $j$, respectively. $P_i$ attempts to count the triangle executing the loop in Line 2-6 with $v = x_1$ and $u = x_2$. However, since $x_2 \notin V_i^c$, $P_i$ sends $N_{x_1}$ to $P_j$ (Line 8). $P_j$ counts this triangle while executing the loop in Line 10-12 with $X = N_{x_1}$, and node $x_3$ appears in $S = N_{x_2} \cap N_{x_1}$ (Line 4 in Fig. \ref{algo:space-tc-msg}). This triangle can never be counted again in any processor, since  $x_1 \notin N_{x_2}$ and $x_1,x_2 \notin N_{x_3}$.
\end{itemize}
Thus, each triangle in $G$ is counted once and only once. $\square$

\subsection{Analysis of the Number of Messages}
\label{sec:msg_spacealgo}
For $v\in V_i^c$, we call $(v,u) \in E$ a \textit{cut edge} if $u \in V_j^c$, $j\neq i$. Let $\ell_{vj}$ is the number of cut edges emanating from node $v$ to all nodes $u$ in partition $j$ with $v\prec u$. Now, in \textit{Surrogate} approach, for all such cut edges $(v,u)$, processor $P_i$ sends $N_v$ to $P_j$ at most once instead of $\ell_{vj}$ times. This leads to a saving of the number of messages by a factor of $\ell_{vj}$ for each $v\in V_i^c$. To get a crude estimate of how the number of messages for \textit{direct} and \textit{surrogate} approaches compare, let $\ell$ be the number of cut edges $\ell_{vj}$ averaged over all $v \in V_i^c$ and partitions $j$. Then, the number of messages exchanged in \textit{direct} approach is roughly $\ell$ larger than \textit{surrogate} approach. 

As shown experimentally in Table \ref{table:msg-comp}, \textit{direct} approach exchanges messages that is $4$ to $12$ times larger than that of \textit{surrogate} approach. Thus,  \textit{surrogate} approach reduces approx. $70\%$ to $90\%$ of messages leading to faster computations as shown in Table \ref{table:op_nop} of the following section.

\begin{table}[!htb]
	\tbl{Number of messages exchanged in Direct and Surrogate approaches.
	\label{table:msg-comp}}{
	\begin{tabular}{|l|l|l|l|} \hline
		\multirow{2}{*}{ \textbf{Networks}} &  \multicolumn{2}{|c|}{\textbf{\# of Messages}} & \multirow{2}{*}{\textbf{$Ratio$}}  \\
		\cline{2-3}
		& Direct & Surrogate &   \\ \hline
		Miami         &    $16,321,478$     & $3,987,871$ & $4.09$     \\ \hline 
		web-Google   &    $493,488$    & $99,221$  & $4.97$   \\ \hline 
		LiveJournal   &    $23,138,824$    & $4,002,575$  &  $5.78$   \\ \hline 
		Twitter   &   $247,821,246$    & $25,341,984$  & $9.78$   \\ \hline 
		PA(10M, 100)         &   $99,436,823$   & $8,092,340$ & $12.29$  \\ \hline 
	\end{tabular}}
\end{table}

\subsection{Experimental Evaluation}
\label{sec:perf}
We presented the experimental evaluation of our algorithm with overlapping partitioning in Section \ref{sec:patric_perf}. In this section, we present the performance of our parallel algorithm with non-overlapping partitioning and compare it with other related algorithms. We will denote our algorithm with overlapping partitioning as \textbf{AOP} and the algorithm with non-overlapping partitioning as \textbf{ANOP} for the convenience of discussion.


\textbf{Comparison with Previous Algorithms.}
Algorithm \textit{AOP} does not require message passing for counting triangles leading to a very fast algorithm (Table \ref{table:op_nop}). In the contrary, \textit{ANOP} achieves huge space saving over \textit{AOP} (Table \ref{table:mem-comp}), although \textit{ANOP} requires message passing for counting triangles.  Our proposed communication approach (surrogate) reduces number of messages quite significantly leading to an almost similar runtime efficiency to that of \textit{AOP}. In fact, \textit{ANOP} loses only $\sim$20\% runtime efficiency for the gain of a significant space efficiency of up to 25 times, thus allowing to work on larger networks.

A runtime comparison among other related algorithms \cite{SURI,park-cikm,park-cikm2} for counting triangles in Twitter network is given in Fig. \ref{fig:twitter_comp}. Our algorithm \textit{ANOP} is $35$, $17$, and $7$ times faster than that of \cite{SURI}, \cite{park-cikm}, and \cite{park-cikm2}, respectively. Further, \textit{ANOP} is almost as fast as \textit{AOP}.

\begin{table}[!hbt]
	\tbl{Runtime performance of our algorithms \textit{AOP} and \textit{ANOP}. We used 200 processors for this experiment. We showed both direct and surrogate approaches for \textit{ANOP}.
	\label{table:op_nop}}{
	\begin{tabular}{|l|l|l|l|l|} \hline
		\multirow{2}{*}{ \textbf{Networks}} &  \multicolumn{3}{|c|}{\textbf{Runtime }} & \multirow{2}{*}{\textbf{Triangles}} \\ 
		\cline{2-4}
		& \textbf{AOP}  & \textbf{Direct} & \textbf{Surrogate}& \\ \hline
		web-BerkStan  &    $0.10$s  & $0.8$s  & $0.14$s & $65$M    \\ \hline 
		Miami         &     $0.6$s    & $3.85$s & $0.79$s & $332$M    \\ \hline 
		LiveJournal   &   $0.8$s     & $5.12$s  & $1.24$s &$286$M   \\ \hline 
		Twitter   &   $9.4$m    &  $35.49$m & $12.33$m  &$34.8$B  \\ \hline 
		PA(1B, 20)         &     $15.5$m    & $78.96$m & $20.77$m &$0.403$M  \\ \hline 
	\end{tabular}}
\end{table}

\begin{figure*}[!tbh]
	\hfill
	\begin{minipage}[t]{.32\textwidth}
		\begin{center}
			\centerline{\includegraphics[width=1.0\textwidth]{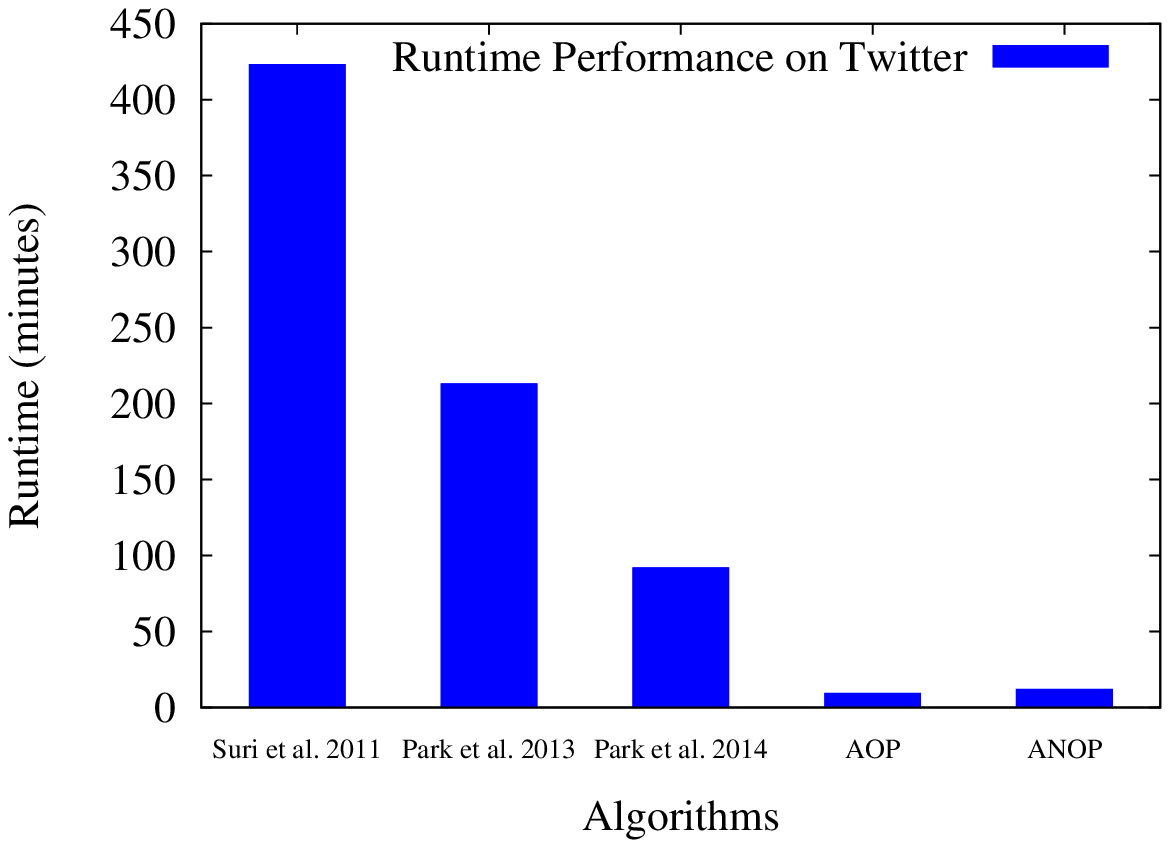}}
			\fsc
			\caption{Runtime reported by various algorithms for counting triangles in Twitter network.}
			\label{fig:twitter_comp}
			\afsc
		\end{center}
	\end{minipage}
	\hfill
	\begin{minipage}[t]{.32\textwidth}
		\begin{center}
			\centerline{\includegraphics[width=1.0\textwidth]{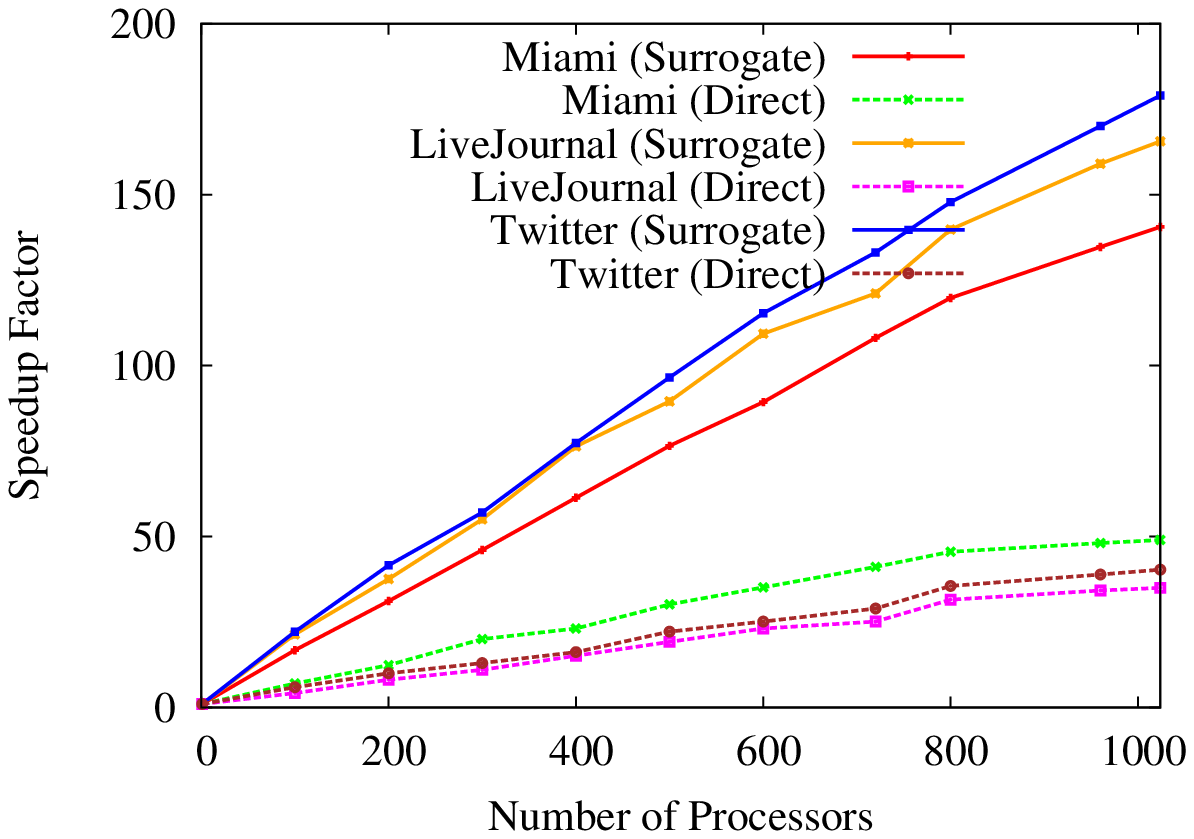}}
			\fsc
			\caption{Speedup factors of our algorithm with both direct and surrogate approaches.}
			\label{fig:strong-space}
			\afsc
		\end{center}
	\end{minipage}
	\hfill
	\begin{minipage}[t]{.32\textwidth}
		\begin{center}
			\centerline{\includegraphics[width=1.0\textwidth]{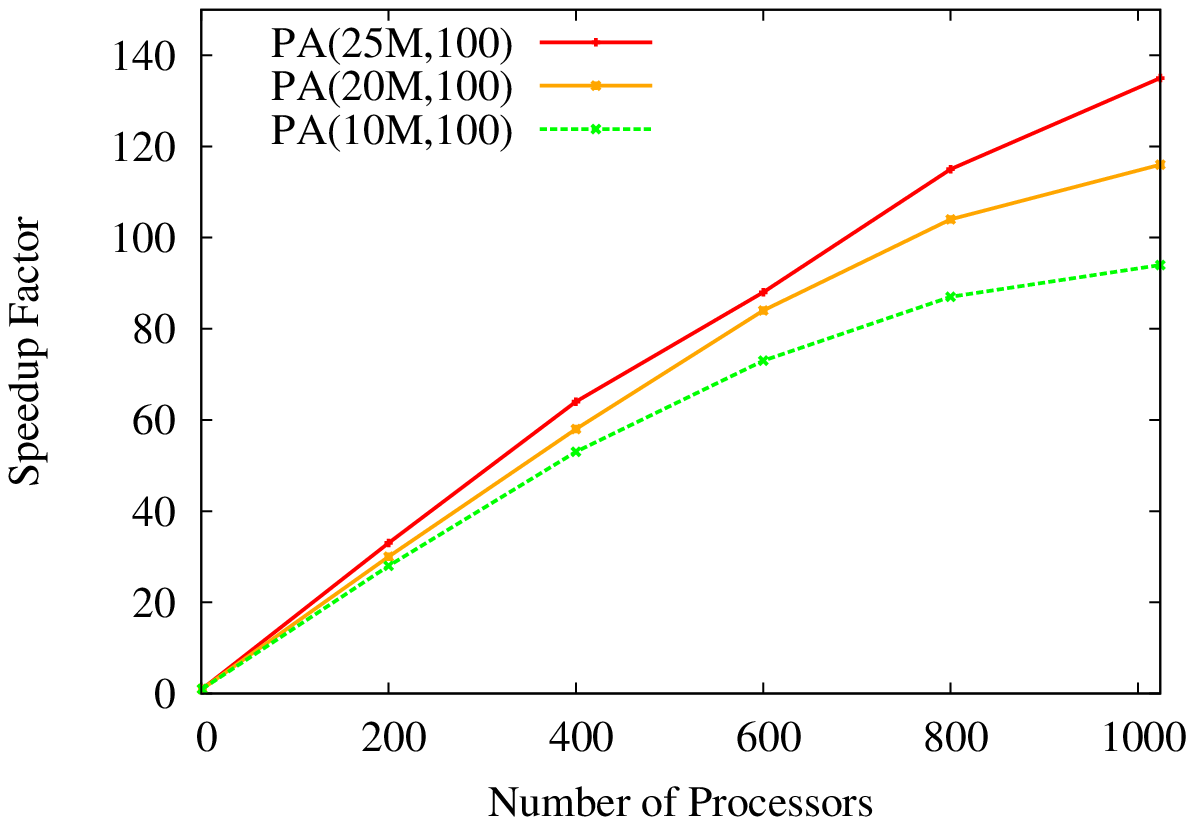}}
			\fsc
			\caption{Improved scalability of our algorithm with increasing network size.}
			\label{fig:space_pscale}
			\afsc
		\end{center}
	\end{minipage}
	\hfill
	\vspace{-5pt}
\end{figure*}

\begin{figure*}[!tbh]
	\hfill
	\begin{minipage}[t]{.32\textwidth}
		\begin{center}
			\centerline{\includegraphics[width=1.0\textwidth]{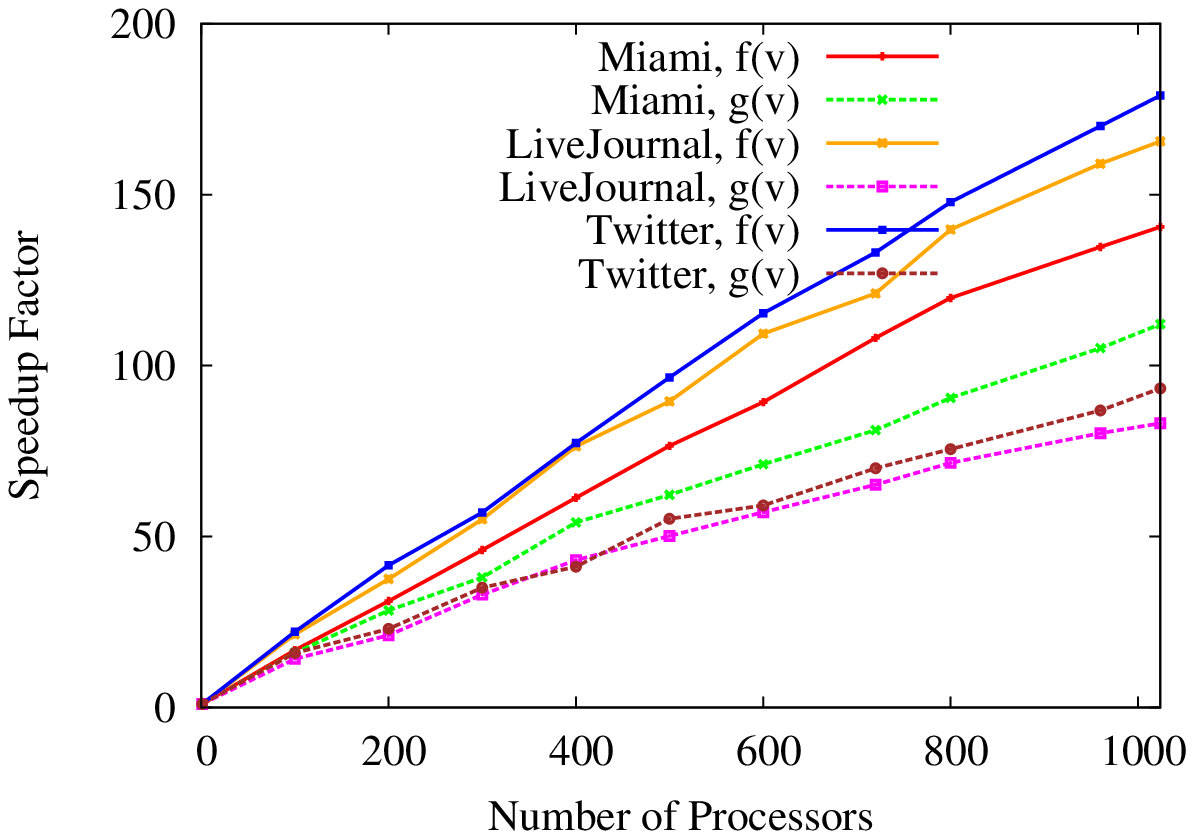}}
			\fsc
			\caption{Comparison of the cost function $f(v)$ estimated for our algorithm with non-overlapping partitioning and the best function $g(v)$ in Section \ref{sec:patric}.}
			\label{fig:new_fv}
			\afsc
		\end{center}
	\end{minipage}
	\hfill
	\begin{minipage}[t]{.32\textwidth}
		\begin{center}
			\centerline{\includegraphics[width=1.0\textwidth]{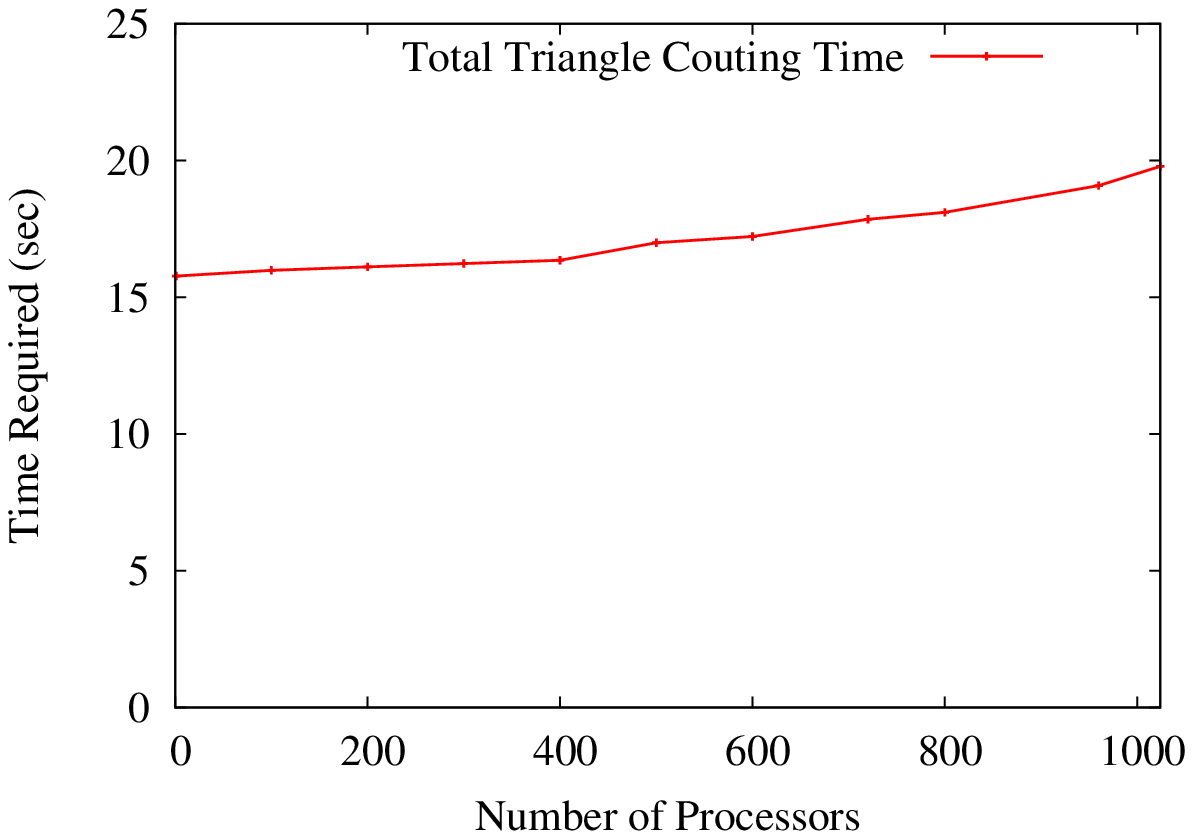}}
			\fsc
			\caption{Weak scaling of our algorithm, experiment performed on PA($t/10*1M, 50$) networks, $t=$ number of processors used.}
			\label{fig:weak-space}
			\afsc
		\end{center}
	\end{minipage}
	\hfill
	\begin{minipage}[t]{.32\textwidth}
		\begin{center}
			\centering
			\scalebox{0.65}{\input{edge_tr.pstex_t}}
			\fsc
			\caption{Two triangles $(v,u,w)$ and $(v',u,w)$ with an overlapping edge $(u,v)$.}
			\label{fig:edge_triangle}
			\afsc
		\end{center}
	\end{minipage}
	
		\hfill
	\vspace{-5pt}
\end{figure*}

\textbf{Strong Scaling.} 
Fig. \ref{fig:strong-space} shows strong scaling (speedup) of our algorithm \textit{ANOP} on Miami, LiveJournal, and web-BerkStan networks with both direct and surrogate approaches. Speedup factors with the surrogate approach are significantly higher than that of the direct approach due to its capability to reduce communication cost drastically. Our algorithm demonstrates an almost linear speedup to a large number of processors.

Further, \textit{ANOP} scales to a higher number of processors when networks grow larger, as shown in Fig. \ref{fig:space_pscale}. This is, in fact, a highly desirable behavior since we need a large number of processors when the network size is large and computation time is high.

\textbf{Effect of Estimations for f(v).} 
We show the performance of our algorithm \textit{ANOP} with the new cost function $f(v) = \sum_{ u\in \mathcal{N}_v-N_v}{(\hat{d_v} + \hat{d_u})}$ and the best function $g(v) = \sum_{ u\in N_v}{(\hat{d_v} + \hat{d_u})}$ computed for \textit{AOP}. As Fig. \ref{fig:new_fv} shows, \textit{ANOP} with $f(v)$ provides better speedup than that with $g(v)$. Function $f(v)$ estimates the computational cost more precisely for \textit{ANOP} with surrogate approach, which leads to improved load balancing and better speedup.

\textbf{Weak Scaling.} 
The weak scaling of our algorithm with non-overlapping partitioning is shown in Fig. \ref{fig:weak-space}. Since the addition of processors causes the overhead for exchanging messages to increase, the runtime of the algorithm increases slowly. However, as the change in runtime is rather slow (not drastic), our algorithm demonstrates a reasonably good weak scaling.

%% file: edge_tr.pstex_t
\begin{picture}(0,0)%
\includegraphics{edge_tr.pstex}%
\end{picture}%
\setlength{\unitlength}{3315sp}%
\begingroup\makeatletter\ifx\SetFigFont\undefined%
\gdef\SetFigFont#1#2#3#4#5{%
  \reset@font\fontsize{#1}{#2pt}%
  \fontfamily{#3}\fontseries{#4}\fontshape{#5}%
  \selectfont}%
\fi\endgroup%
\begin{picture}(3585,2125)(1516,-3941)
\put(1531,-2851){\makebox(0,0)[lb]{\smash{{\SetFigFont{17}{20.4}{\rmdefault}{\mddefault}{\updefault}{\color[rgb]{0,0,0}$v$}%
}}}}
\put(5086,-2806){\makebox(0,0)[lb]{\smash{{\SetFigFont{17}{20.4}{\rmdefault}{\mddefault}{\updefault}{\color[rgb]{0,0,0}$v'$}%
}}}}
\put(2926,-1996){\makebox(0,0)[lb]{\smash{{\SetFigFont{17}{20.4}{\rmdefault}{\mddefault}{\updefault}{\color[rgb]{0,0,0}$u$}%
}}}}
\put(3826,-3886){\makebox(0,0)[lb]{\smash{{\SetFigFont{17}{20.4}{\rmdefault}{\mddefault}{\updefault}{\color[rgb]{0,0,0}$w$}%
}}}}
\end{picture}%

%% file: sparsification.tex

\section{Sparsification-based Parallel Approximation Algorithms}
\label{sec:sparse}

We discussed our parallel algorithms for counting the exact number of triangles in Section \ref{sec:patric} and \ref{sec:space}. In this section, we show how those algorithms can be combined with an edge sparsification technique to design a parallel approximation algorithm.  

Sparsification of a network is a sampling technique where some randomly chosen edges are retained and the rest are deleted, and then computation is performed on the sparsified network. Such technique saves both computation time and memory space and provides an approximate result. We integrate a sparsification technique, called DOULION, proposed in \cite{tsourakakis2009doulion} with our parallel algorithms. Our adapted version of DOULION provides more accuracy than DOULION when used with overlapping partitioning. The adaptation with non-overlapping partitioning provides the same accuracy as original DOULION. 

\subsection{Overview of the Sparsification}
Let $G(V,E)$ and $G'(V,E' \subset E)$  be the networks before and after sparsification, respectively. Network $G'(V,E')$ is obtained from $G(V,E)$ by retaining each edge,  independently, with probability $q$ and removing with probability $1-q$.  Now any algorithm can be used to find the exact number of triangles in $G'$. Let $T(G')$ be the number of triangles in $G'$. The estimated number of triangles in $G$ is given by $\frac{1}{q^3}T(G')$, which is an unbiased estimation. It is easy to see that the expected value of  $\frac{1}{q^3}T(G')$ is T(G), the number of triangles in the original network $G$: let the triangles in $G$ be arbitrarily numbered as $1,2, \ldots , T(G)$, and $x_i$ be an indicator random variable that takes value 1 if triangle $i$ of $G$ survives in  $G'$. A triangle survives if all of its three edges are retained in $G'$. Then we have $\Pr\{x_i=1\} = \frac{1}{q^3}$ and, by the linearity of expectation,
\vspace{-.1in}
\begin{equation*} \label{eqn:exp}
{\small
E\left[\frac{1}{q^3}T(G')\right] = \frac{1}{q^3}\sum_{i=1}^{T(G)}{E[x_i]} = \frac{1}{q^3}\sum_{i=1}^{T(G)}{\Pr\{x_i = 1\}} = T(G).
}
\end{equation*}

As shown in \cite{tsourakakis2009doulion}, the variance of the estimated number of triangles is
\begin{equation} \label{eqn:var1}
\mbox{ Var } = \left(\frac{1}{q^3} - 1\right) T(G) + 2k\left(\frac{1}{q} - 1\right),
\end{equation}
where $k$ is the number of pairs of triangles in $G$ with an overlapping edge (see Figure \ref{fig:edge_triangle}).

\subsection{Parallel Sparsification Algorithm}
In our parallel algorithm, sparsification is done as follows: each processor $P_i$ independently performs sparsification on its partition $G_i(V'_{i},E'_{i})$, where $V'_{i}=V_i$, $E'_{i}=E_i$ for AOP and $V'_{i}=V_i^c$, $E'_{i}=E_i^c$ for ANOP. While loading the partition $G_i$ into its local memory, $P_i$ retains each edge $(u,v) \in E'_i$ with probability $q$ and discards it with probability $1-q$ as shown Figure \ref{algo:sparsification}. 

\begin{figure}[!h]
\begin{center}
\fbox{
\begin{minipage}[c] {0.95\linewidth}
\begin{algorithmic}[1]
\FOR {$v \in V'_{i}$} 
	\FOR {$(v,u) \in E'_{i}$}
		\IF { $v \prec u $}
				\STATE store $u$ in $N_v$ with probability $q$
		\ENDIF
	\ENDFOR
\ENDFOR
\STATE $T_i\leftarrow$ count of triangles on $G_i$ //using alg. in Sec. \ref{sec:patric} or \ref{sec:space}
\STATE Find Sum $T'=\sum_i{T_i}$ using $\textsc{MpiReduce}$
\STATE $T \leftarrow \frac{1}{q^3} \times T'$
\end{algorithmic}
\end{minipage}
}
\end{center}
\fsc
\caption{Counting approximate number of triangles with parallel sparsification algorithm.}
\label{algo:sparsification}
\end{figure}

 
Now, our parallel sparsification algorithm with overlapping partitioning is not exactly the same as that of DOULION. Consider two triangles $(v,u,w)$ and $(v',u,w)$ with an overlapping edge $(u,w)$ as shown in Fig. \ref{fig:edge_triangle}. In DOULION, if edge $(u,w)$ is not retained, none of the two triangles survive, and as a result, survivals of $(v,u,w)$ and $(v',u,w)$ are not independent events. Now, in our case, if $v$ and $v'$ are core nodes in two different partitions $G_i$ and $G_j$, processor $i$ may retain edge $(u,w)$ while processor $j$ discards $(u,w)$, and vice versa. As processor $i$ and $j$ perform sparsification independently, survivals of triangles $(v,u,w)$ and $(v',u,w)$ are independent events.

However, our estimation is also unbiased, and in fact, this difference (with DOULION) improves the accuracy of the estimation by our parallel algorithm. Since the probability of survival of any triangle is still exactly $\frac{1}{q^3}$, we have $E\left[\frac{1}{q^3}T'\right] = T$. To calculate variance of the estimation, let $k'_i$ be the number of pairs of triangles with an overlapping edge such that both triangles are in partition $G_i$, and $k' = \sum_i{k'_i}$. Let $k''$ be the number of pairs of triangles $(v,u,w)$ and $(v',u,w)$ with an overlapping edge $(u,w)$ (as shown in Fig. \ref{fig:edge_triangle}) and $v$ and $v'$ are core nodes in two different partitions. Then clearly, $k' + k'' = k$ and $k' \le k$. Now following the same steps as in \cite{tsourakakis2009doulion}, one can show that the variance of our estimation is
\begin{equation} \label{eqn:var2}
\mbox{ Var}' = \left(\frac{1}{q^3} - 1\right) T(G) + 2k'\left(\frac{1}{q} - 1\right).
\end{equation}

Comparing Eqn. \ref{eqn:var1} and \ref{eqn:var2}, if $k'' > 0$, we have $k' < k$ and reduced variance leading to improved accuracy. We verify this observation by the experimental results on one realistic synthetic and three real-world networks in Table \ref{vs_doulion}. For all networks, our parallel sparsification algorithm with overlapping partitioning results in smaller variance and errors than that of DOULION. 

However, the accuracy does not improve for parallel sparsification with non-overlapping partitioning. Since the partitioning is non-overlapping, the effect of parallel sparsification is the same as that of the sequential sparsification. As a result, our parallel sparsification algorithm with non-overlapping partition has effectively the same accuracy as that of DOULION, as evident in Table \ref{vs_doulion}.

\begin{table*}
		\tbl{Accuracy of our parallel sparsification algorithm and DOULION \cite{tsourakakis2009doulion} with $q = 0.1$. Our parallel algorithm was run with 100 processors. Variance, max error and average error are calculated from 25 independent runs for each of the algorithms. The best values for each attribute are marked as bold.
		\label{vs_doulion}}{
		\begin{tabular}{|l|l|l|l|l|l|l|l|l|l|} \hline
			\multirow{2}{*}{\textbf{Networks}} & \multicolumn{3}{|c|}{\textbf{Variance}} & \multicolumn{3}{|c|}{\textbf{Avg. error (\%)}} & \multicolumn{3}{|c|}{\textbf{Max error (\%)}}  \\
			\cline{2-10}
			& AOP & ANOP & DOULION & AOP & ANOP & DOULION & AOP & ANOP & DOULION   \\ \hline
			web-BerkStan &  \textbf{1.287} &  1.991  &  2.027   &        \textbf{0.389}  &  0.391    &  0.392     &     \textbf{1.024} & 1.082  &  1.082   \\ \hline
			LiveJournal  &  \textbf{1.770} &   1.952    &  1.958   &       \textbf{1.463}  &  1.857     &  1.862     &    \textbf{3.881} &  4.774   &  4.752    \\ \hline
			web-Google &  \textbf{1.411}  & 2.003   & 1.998    & \textbf{1.327 }       & 1.564     & 1.580      & \textbf{2.455}     & 3.923 & 3.942   \\ \hline
			Miami  & \textbf{1.675}  & 2.105    & 2.112   & \textbf{1.55}        &  1.921   &  1.905    & \textbf{3.45}   & 4.88   &  4.75   \\ \hline
		\end{tabular}}
\end{table*}

\begin{table*}
		\tbl{Comparison of accuracy between our parallel sparsification algorithms and DOULION on one realistic synthetic and three real-world networks with 100 processors. The best values for each $q$ are marked as bold.
		\label{table:sparse_accuracy}}{
		\begin{tabular}{| l | l | l | l | l | l | l |}
			\hline
			\textbf{Networks} & \textbf{Algorithms} &  \textbf{$q=0.1$} & \textbf{$q=0.2$} & \textbf{$q=0.3$} & \textbf{$q=0.4$} & \textbf{$q=0.5$} \\ \hline
			\multirow{3}{*}{\textbf{web-BerkStan}} & AOP   &  \textbf{99.9921}  &  \textbf{99.9927} & \textbf{99.9932} & \textbf{99.9947} & $\textbf{99.9979}$\\
			\cline{2-7}
			& ANOP   &  $99.6308$  &  $99.7490$ & $99.8392$ & $99.9168$ & $99.9565$\\
			\cline{2-7}
			& DOULION    &  $99.6309$  &  $99.7484$ & $99.8401$ & $99.9171$ & $99.9566$\\
			\hline
			\multirow{3}{*}{\textbf{LiveJournal}} & AOP   &  $\textbf{99.9914}$  &  $\textbf{99.9917}$ & $\textbf{99.9924}$ & $\textbf{99.9936}$ & $\textbf{99.9971}$\\
			\cline{2-7}
			& ANOP   &  $99.6325$  &  $99.7488$ & $99.8412$ & $99.9178$ & $99.9575$\\
			\cline{2-7}
			& DOULION    &  $99.6310$  &  $99.7544$ & $99.8392$ & $99.9121$ & $99.9584$\\
			\hline
			\multirow{3}{*}{\textbf{web-Google}} & AOP   &  $\textbf{99.9917}$  &  $\textbf{99.9923}$ & $\textbf{99.9929}$ & $\textbf{99.9939}$ & $\textbf{99.9975}$\\
			\cline{2-7}
			& ANOP   &  $99.6299$  &  $99.7391$ & $99.8435$ & $99.9168$ & $99.9577$\\
			\cline{2-7}
			& DOULION    &  $99.6305$  &  $99.7398$ & $99.8428$ & $99.9170$ & $99.9574$\\
			\hline
			\multirow{3}{*}{\textbf{Miami}} & AOP   &  $\textbf{99.9916}$  &  $\textbf{99.9919}$ & $\textbf{99.9926}$ & $\textbf{99.9938}$ & $\textbf{99.9974}$\\
			\cline{2-7}
			& ANOP   &  $99.6285$  &  $99.7495$ & $99.8384$ & $99.9168$ & $99.9562$\\
			\cline{2-7}
			& DOULION    &  $99.6288$  &  $99.7494$ & $99.8381$ & $99.9169$ & $99.9563$\\
			\hline
		
		\end{tabular}}
		
\end{table*}

Sparsification reduces memory requirement since only a subset of the edges are stored in the main memory. As a result, adaptation of sparsification allows our parallel algorithms to work with even larger networks. With sampling probability $q$ (the probability of retaining an edge), the expected number of edges to be stored in the main memory is $q|E|$. Thus, we can expect that the use of sparsification with our parallel algorithms will allow us to work with a network $1/q$ times larger. Sparsification technique also offers additional speedup due to working on a reduced graph. In \cite{tsourakakis2009doulion}, it was shown that due to sparsification with parameter $q$, the computation can be faster as much as $1/p^2$ times. However, in practice the speed up is typically smaller than $1/p^2$ but larger than $1/p$. As an example, with our parallel sparsification with AOP on LiveJournal network, we obtain speedups of $57.88$, $24.36$, $11.04$, $6.19$, and $4.0$ for $q=0.1$ to $0.5$, respectively. When an application requires only an approximate count of the total triangles in graph with a reasonable accuracy, such parallel sparsification algorithm will be proven useful.

%% file: trianglelisting.tex
\section{Listing Triangles in Graphs} \label{sec:trianglelisting}

Our parallel algorithms for counting triangles in Section \ref{sec:patric} and \ref{sec:space} can easily be extended to list all triangles in graphs. Triangle listing has various  applications in the analysis of graphs such as the computation of clustering coefficients, transitivity, triangular connectivity, and trusses \cite{chu11triangle}. Our parallel algorithms counts the exact number of triangles in the graph. To count the number of triangles incident on an edge $(u,v)$, the algorithms perform a set intersection operation $N_v\cap N_u$. After each intersection operation, all associated triangles can be listed simply by the code shown in Fig. \ref{algo:tlisting}.

\begin{figure}[!ht]
	\begin{center}
		\fbox{
			\begin{minipage}[c] {0.9\linewidth}
				\begin{algorithmic}[1]
					\STATE $S \leftarrow N_v \cap N_u$
					\FOR {$w \in S$}
					\STATE \textit{Output} triangle $(u,v,w)$
					\ENDFOR
				\end{algorithmic}
			\end{minipage}
		}
	\end{center}
	\caption{Listing triangles after performing the set intersection operation for counting triangles.}
	\label{algo:tlisting}
\end{figure}

%% file: clusteringco.tex

\section{Computing Clustering Coefficient of Nodes} \label{sec:clustering}

Our parallel algorithms can be extended to compute local clustering coefficient without increasing the cost significantly. In a sequential setting, an algorithm for counting triangles can be directly used for computing clustering coefficients of the nodes by simply keeping the counts of triangles for each node individually. However, in a distributed-memory parallel system, combining the counts from all processors for a node poses another level of difficulty. We present an efficient aggregation scheme for combining the counts for a node from different processors.

\textbf{Parallel Computation of Clustering Coefficients.} Recall that clustering coefficients of nodes $v$ is computed as follows:
\begin{eqnarray*}
	C_v = \frac{T_v}{{d_v \choose 2}} = \frac{2T_v}{d_v(d_v - 1)},
\end{eqnarray*}
where $T_v$ is the number of triangles containing node $v$.

Our parallel algorithms for counting triangles count each triangle only once. However, all triangles containing a node $v$ might not be computed by a single processor. Consider a triangle $(u,v,w)$ with $u \prec_{\mathcal{D}} v \prec_{\mathcal{D}} w$. Further, assume that $u \in V_i^c$, $v \in V_j^c$, and $w \in V_k^c$, where $i \ne j \ne k$. Now, for our parallel algorithm \textit{AOP}, the triangle $(u,v,w)$ is counted by $P_i$. Let $T_v^i$ be the number of triangles incident on node $v$ computed by $P_i$. We also call such counts \textit{local counts} of $v$ in processor $P_i$. For the triangle $(u,v,w)$, $P_i$ tracks local counts of all of $u$, $v$, and $w$. Thus, the total count of triangles incident on a node $v$ might be distributed among multiple processors. Each processor $P_i$ needs to aggregate local counts of $u\in V_i^c$ from other processors. (For algorithm \textit{ANOP}, the above triangle $(u,v,w)$ is counted by $P_j$, and a similar argument as above holds.)

To aggregate local counts from other processors, the following approach can be adopted: for each processor, we can store local counts $T_v^i$ in an array of size $\Theta(n)$ and then use MPI \textit{All-Reduce } function for the aggregation. However, for a large network, the required system buffer to perform MPI aggregation on arrays of size $\Theta(n)$ might be prohibitive. Another approach for aggregation might be as follows. Instead of using main memory, local counts can be written to disk files based on some hash functions of nodes. Each processor $P_i$ then aggregates counts for nodes $v \in V_i^c$ from $P$ disk files. Even though this scheme saves the usage of main memory, performing a large number of disk I/O leads to a large runtime.  

Both of the above approach compromises either the runtime or space efficiency. We use the following approach which is both time and space efficient. 

Our approach involves two steps. First, for each triangle counted by $P_i$, it tracks local counts $T.^i$ as shown in Figure \ref{algo:localcount}.

\begin{figure}[!ht]
	\begin{center}
		\fbox{
			\begin{minipage}[c] {0.9\linewidth}
				\begin{algorithmic}[1]
					\FOR {for each triangle $(v,u,w)$ counted in $G_i$}
					\STATE $T_v^i \leftarrow T_v^i + 1$
					\STATE $T_u^i \leftarrow T_u^i + 1$
					\STATE $T_w^i \leftarrow T_w^i+1$
					\ENDFOR
				\end{algorithmic}
			\end{minipage}
		}
	\end{center}
	\caption{Tracking local counts by processor $P_i$. Each triangle $(v,u,w)$ is detected by the triangle listing algorithm shown in Fig. \ref{algo:tlisting}.}
	\label{algo:localcount}
\end{figure}

Second, processor $P_i$ aggregates local counts of nodes $v \in V_i^c$ from other processors. Total number of triangles $T_v$ incident on $v$ is given by $T_v = \sum_{j\ne i} T_v^j$. Each processor $P_j$ sends local counts $T_v^j$ of nodes $v \in V_i^c$ encountered in any triangles counted in partition $j$. $P_i$ receives those counts and aggregates to $T_v$. We present the pseudocode of this aggregation in Figure \ref{algo:cc}. Finally, $P_i$ computes $C_v = \frac{2T_v}{d_v(d_v -1)}$ for each $v \in V_{i}^c$.

\begin{figure}[!ht]
	\begin{center}
		\fbox{
			\begin{minipage}[c] {0.9\linewidth}
				\begin{algorithmic}[1]
					\FOR {$v \in V_i^c$}
					\STATE $T_v \leftarrow T_v^i$
					\ENDFOR
					
					\FOR {each processor $P_j$}
					\STATE Construct \textit{message} $\langle Y_i^j, \mathcal{T}_i^j \rangle$ s.t.:\\
					$Y_i^j \leftarrow \{v| v\in N_u, u\in V_i^c\}  \cap V_j^c$, $\mathcal{T}_i^j \leftarrow \{T_v^i | v \in Y_i^j\}$.
					\STATE  Send \textit{message} $\langle Y_i^j, \mathcal{T}_i^j \rangle$ to $P_j$
					\ENDFOR
					
					\FOR {each processor $P_j$}
					\STATE  Receive \textit{message} $\langle Y_j^i, \mathcal{T}_j^i \rangle$ from $P_j$
					\STATE $T_v \leftarrow T_v + T_v^j$ 
					\ENDFOR
				\end{algorithmic}
			\end{minipage}
		}
	\end{center}
	\caption{Aggregating local counts for $v\in V_i^c$ by $P_i$.}
	\afsc
	\label{algo:cc}
\end{figure}

Our approach tracks local counts for nodes $v \in V_i^c$ and neighbors of such $v$ which requires, in practice, significantly smaller than $\Theta(n)$ space. Next, we show the performance of our algorithm.

\textbf{Performance.}  We show the strong and weak scaling of our algorithm for computing clustering coefficients of nodes in Fig. ~\ref{fig:speedup_cc} and \ref{fig:combined_weak}, respectively. The algorithm shows good speedups and scales almost linearly to a large number of processors. Since aggregating local counts introduces additional inter-processor communication, the speedups are a little smaller than that of the triangle counting algorithms. For the same reason, the weak scalability of the algorithm is a little smaller than that of the triangle counting algorithms. However, the increase of runtime with additional processors is still not drastic, and the algorithm shows a good weak scaling. 

\begin{figure*}[!tbh]
	\hfill
	\begin{minipage}[t]{.47\textwidth}
		\begin{center}
			\centering
			\includegraphics[scale=0.5]{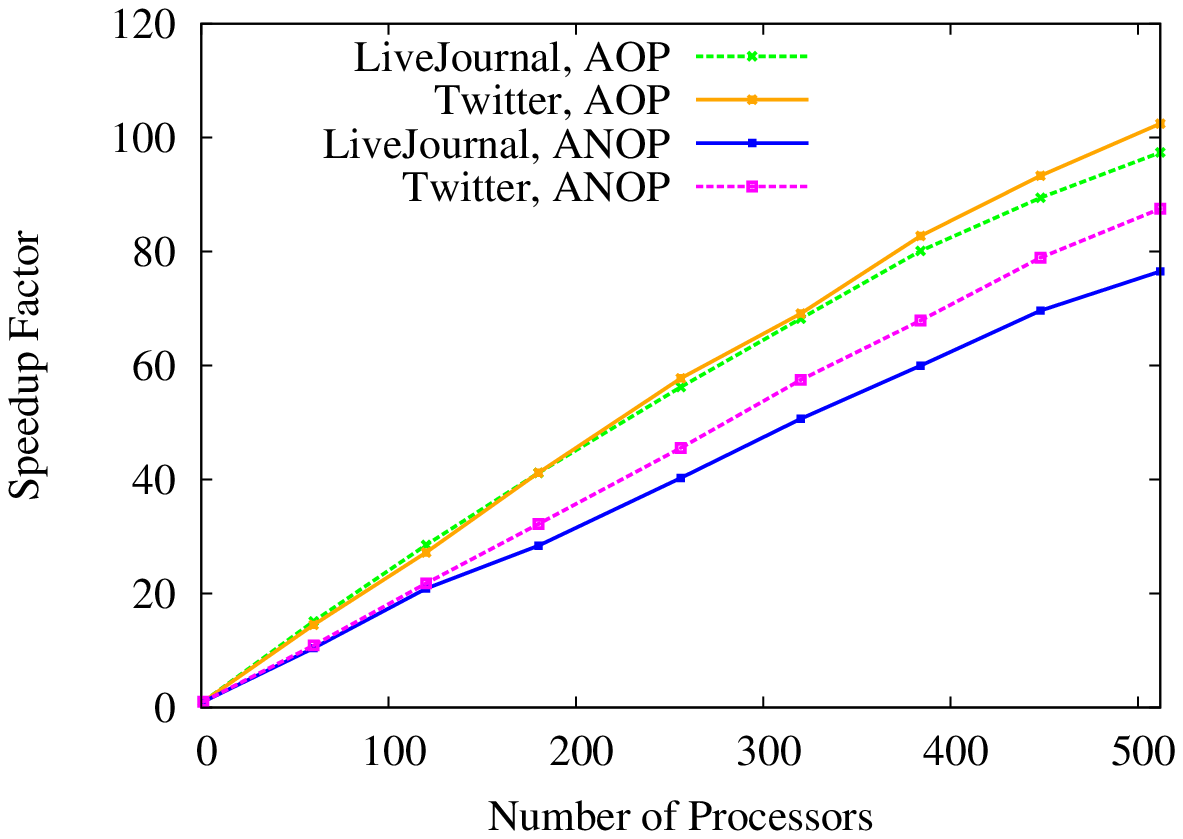}
			\caption{Strong scaling of clustering coefficient algorithm with both \textit{AOP} and \textit{ANOP} on LiveJournal and Twitter networks.}
			\label{fig:speedup_cc}
		\end{center}
	\end{minipage}
	\hfill
	\begin{minipage}[t]{.47\textwidth}
		\begin{center}
			\centering
			\includegraphics[scale=0.5]{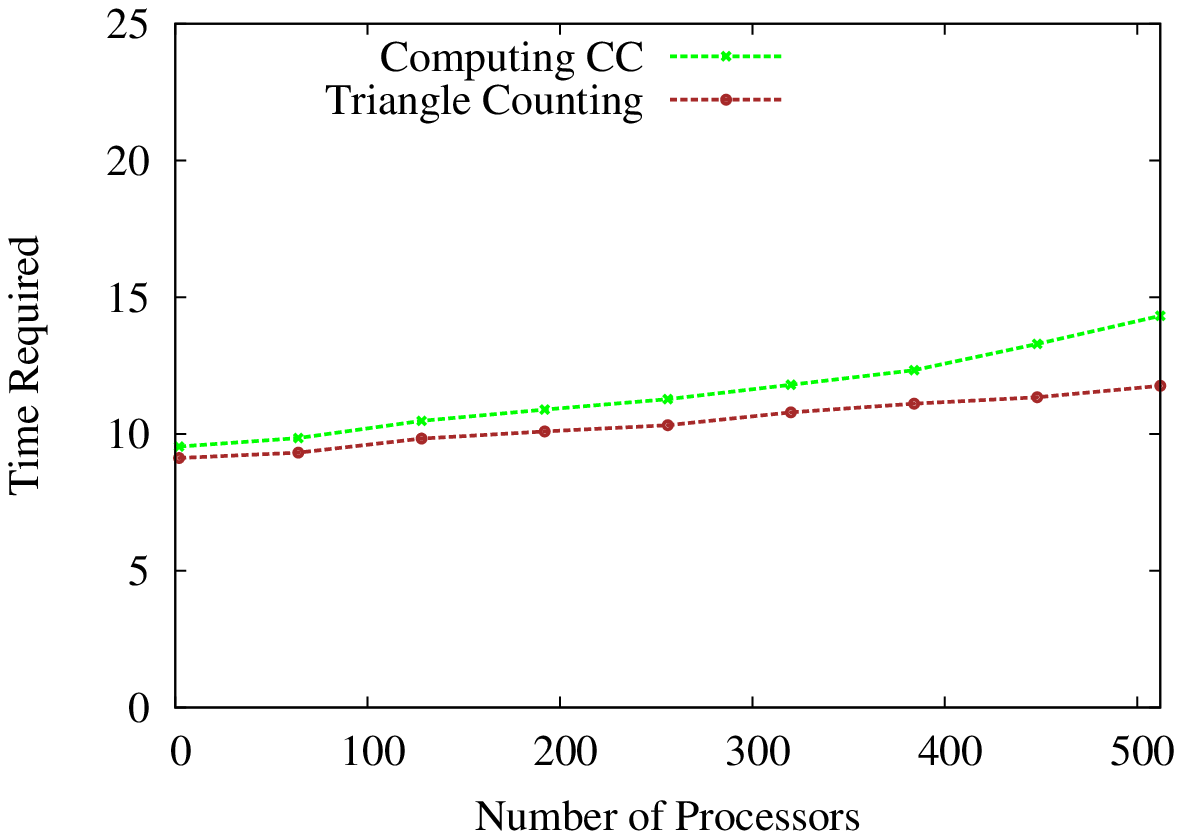}
			\caption{Weak scaling of the algorithms for computing clustering coefficient (CC) and counting triangles (TC).}
			\label{fig:combined_weak}
		\end{center}
	\end{minipage}
	\hfill
	\vspace{-5pt}
\end{figure*}

%% file: application.tex
\section{Applications for Counting Triangles} \label{sec:applications}

The number of triangles in graphs have many important applications in data mining. Becchetti et al. \cite{BECC} showed how the number of triangles can be used to detect spamming activity in web graphs. They used a public web spam dataset and compared it with a non-spam dataset: first, they computed the number of triangles for each host and plotted the distribution of triangles and clustering coefficients for both dataset. Using Kolmogorov-Smirnov test, they concluded the distributions are significantly different for spam and non-spam datasets. Further, the authors also showed how to comment on the role of individual nodes in a social network based on the number of triangles they participate. Eckmann et al. \cite{eckmann2002theme} used triangle counting in uncovering the thematic structure of the web. The abundance of triangles also implies community structures in graphs. Nodes forming a subgraph of high triangular density usually belong to the same community. In fact, the number of triangles incident on nodes has been used by several methods in the literature of community detection \cite{PratPerez2016PTT,Zhang2009PCD,Soman2011FCD}. The computation of clustering coefficients also requires the number of triangles incident on nodes. Social networks usually demonstrate high average clustering coefficients. We show how clustering coefficients can be computed using our parallel algorithms in Section \ref{sec:clustering}. 

In this section, we discuss how the number of triangles can be used to characterize various types of networks. There is a multitude of real-world networks including social contact networks, online social networks, web graphs, and collaboration networks. These networks vary in terms of triangular density and community or social structure in them. As a result, it is possible to characterize real-world networks based on their triangle based statistics. We define the \textit{normalized triangle count} (\textbf{NTC}) as the mean number of triangles per node in the network. We compute NTC for a variety of networks and show the comparison in Table \ref{table:triangle-cc}. Many random graph models such as Erd\H{o}s-R\'{e}yni and Preferential Attachment models do not generate many triangles, and the resulting NTCs are also very low. Some communication and web graphs (e.g., Email-Enron) generate a descent number of triangles because of the nature of the communication and links among web pages in the host domain. When social or cluster structure exists in the network, we get a larger number of triangles per node, as shown in Table \ref{table:triangle-cc} for LiveJournal and web-BerkStan networks. Further, for networks with a more developed social structure and realistic person-to-person interactions, NTCs are very large, as evident for Miami, com-Orkut, and Twitter networks. Thus the number of triangles offers good insights about the underlying social and community structures in networks. 

\begin{table*}
	\tbl{ Comparison of the number of triangles ($\triangle$) and normalized triangle count (NTC) in various networks. We used both artificially generated and real-world networks.
	\label{table:triangle-cc}}{
	\begin{tabular}{ | l | c | c | c |}
		\hline
		{\bf Network} & {\bf $n$} & {\bf $\triangle$} & {\bf NTC$(\triangle/n)$} \\ \hline
		Gnp$(500K,20)$   &  $500$K    &   $1308$  & $0.0026$  \\ \hline
		PA$(25M,50)$ &	$25$M	& $1.3$M &	$0.052$ \\ \hline \hline
		Email-Enron   &  $37$K    &   $727044$  &  $19.815$  \\ \hline
		web-Google	&  $0.88$M	& $13.39$M	& $15.293$	 \\ \hline \hline
		LiveJournal	& $4.85$M	& $285.7$M	& $58.943$	 \\ \hline
		web-BerkStan  &  $0.69$M    &   $64.69$M & $94.408$  \\ \hline \hline
		Miami         &  $2.1$M    &   $332$M    &$158.095$ \\ \hline
		com-Orkut  &  $3.07$M    &   $628$M    &$204.262$  \\ \hline
		Twitter       &  $42$M   &     $34.8$B & $828.571$  \\ \hline
	\end{tabular}}
\end{table*}

%% file: conclusion.tex
\section{Conclusion}
\label{sec:conclusion}
We presented parallel algorithms for counting triangles and computing clustering coefficients in massive networks. These algorithms can work with networks that have billions of nodes and edges. Such capability of our algorithms will enable various types of analysis of massive real-world
networks, networks that otherwise do not fit in the main memory of a single computing node. These algorithms show very good scalability with both the number of processors and the problem size and performs well on both real-world and artificial networks. We have been able to count
triangles of a massive network with $10B$ edges in less than $16$ minutes. We presented several load balancing schemes and showed that such schemes provide very good balancing. Further, we have adopted the sparsification approach of DOULION in our parallel algorithms with improved accuracy. This adoption will allow us to deal with even larger networks. We also extend our triangle counting algorithm for listing triangles and computing clustering coefficients in massive graphs.